\documentclass[a4paper,reqno]{amsart}

\usepackage{amsfonts,amssymb,amsmath,amsthm,gensymb,color,epic,eepic,float,fullpage}
\usepackage[mathscr]{euscript}
\usepackage[T1]{fontenc}
\usepackage{mathtools}
\let\savedegree\degree
\let\degree\relax
\usepackage{mathabx}
\let\degree\savedegree
\usepackage{comment}
\usepackage{stmaryrd}
\usepackage{hyperref}
\usepackage{dsfont}
\usepackage{bbm}
\usepackage{enumitem}
\usepackage{cite}

\usepackage{tikz}
\usetikzlibrary{graphs,patterns,decorations.markings,arrows,matrix}
\usetikzlibrary{calc,decorations.pathmorphing,decorations.markings,
decorations.pathreplacing,patterns,shapes,arrows}
\tikzstyle{block} = [rectangle,draw,text width=10em,text centered,rounded corners,minimum height=4em]
\tikzstyle{line} = [draw, -latex']
\tikzset{arrow/.style={postaction={decorate,thick,decoration={markings,mark = at position #1 with {\arrow{>}}}}},arrow/.default=1}
\usepackage{youngtab}
\allowdisplaybreaks
\newtheorem{defn}{Definition}
\newtheorem{lem}{Lemma}
\newtheorem{thm}{Theorem}

\newtheorem{cor}{Corollary}

\newtheorem{prop}{Proposition}

\def\bbZ{\mathbb{Z}}
\def\bbC{\mathbb{C}}

\def\leq{\leqslant}
\def\geq{\geqslant}

\def\a{\alpha}
\def\b{\beta}
\def\e{\epsilon}
\def\ve{\varepsilon}
\def\g{\gamma}
\def\D{\Delta}
\def\d{\delta}
\def\k{\kappa}

\def\l{\lambda}

\def\m{\mu}
\def\mb{\bar{\mu}}
\def\n{\nu}
\def\nb{\bar{\nu}}
\def\s{\sigma}
\def\r{\rho}
\def\t{\tau}

\def\th{\theta}
\def\vphi{\varphi}

\def\MR{\mathcal{R}}
\def\DR{\mathds{R}}
\def\MH{\mathcal{H}}

\def\MN{\mathcal{N}}
\def\MRb{\widebar{\mathcal{R}}}
\def\Rb{\widebar{R}}

\def\MP{\mathcal{P}}

\def\MB{\mathcal{B}}
\def\MU{\mathcal{U}}
\def\MK{\mathcal{K}}
\def\MP{\mathcal{P}}
\def\MH{\mathcal{H}}

\def\psub{\setminus}
\def\padd{\cup}
\def\pint{\cap}
\def\pin{\subseteq}
\def\pnin{\nsubseteq}
\def\pzero{\emptyset}

\def\Cp{C^{\perp}}
\def\cp{c^{\perp}}
\def\Dp{D^{\perp}}
\def\Dt{\widetilde{\D}}

\def\eb{\bar{\e}}

\def\Pit{\widetilde{\Pi}}
\def\Eb{\bar{E}}

\def\ve{\varepsilon}

\def\ed{\bar{e}}
\def\fd{\bar{f}}
\def\psid{\bar{\psi}}

\def\deg{\text{deg}}
\def\pdeg{\text{pdeg}}
\def\hdeg{\text{hdeg}}

\def\id{\text{id}}

\def\at{\tilde{a}}

\def\qq{{\sf q}}
\def\tt{{\sf  t}}

\def\dd{\text{d}}

\newcommand{\braket}[2]{\left\langle #1|#2\right\rangle}
\newcommand{\bra}[1]{\left\langle #1\right|}
\newcommand{\ket}[1]{\left|#1\right\rangle}
\newcommand{\qg}{U_{q,t}(\overset{..}{gl}_1)}
\newcommand{\sln}{U_{q}(\widehat{sl_n})}

\newcommand{\ba}{\[\begin{aligned}~}
\newcommand{\ea}{\end{aligned}\]}
\def\gt{\mathfrak{g}}

\begin{document}

\title[]{The $R$-matrix of the quantum toroidal algebra $\qg$ in the Fock module}

\author{Alexandr Garbali and Jan de Gier}
\address{ARC Centre of Excellence for Mathematical and Statistical Frontiers (ACEMS), School of Mathematics and Statistics, University of Melbourne, Parkville, Victoria 3010, Australia}
\email{alexandr.garbali@unimelb.edu.au, jdgier@unimelb.edu.au}

\begin{abstract}
We propose a method to compute the $R$-matrix $R$ on a tensor product of Fock modules from coproduct relations in a Hopf algebra. We apply this method to the quantum toroidal algebra $\qg$. We show that the coproduct relations of $\qg$ reduce to a single elegant equation for $R$. Using the theory of symmetric Macdonald polynomials we show that this equation provides a recursive formula for the matrix elements of $R$.
\end{abstract}

\maketitle

\section{Introduction}
The quantum toroidal algebra $\gt=\qg$\footnote{This algebra is also known under the names: Ding--Iohara--Miki algebra (DIM), quantum continuous $gl_\infty$, $(q,\g)$ analog of the $W_{1+\infty}$ algebra and the elliptic Hall algebra.} has an elegant representation theory related to combinatorics of partitions and plane partitions and the theory of symmetric Macdonald polynomials \cite{Berg,BurbS12,SchiffV13,Miki,FHSSY,FT_Fock,FFJMM_F,FJMM_PP,FHHSY}. It received recent attention in relation to algebraic geometry \cite{FT_Fock,Smirn16,MOk,Neg15,OkSm}, integrable models \cite{FJMM_BA,FJMM_FTM,FJMM_IM} and the AGT correspondence \cite{AFHKSY,SchiffV13a,NegAGT}. The algebra $\qg$ in the Fock representation\footnote{In the literature on the topological vertex associated to $\gt$ the representation which we refer to as the Fock representation is known as the horizontal representation. This representation corresponds to the case of non-commutative Cartan currents in the Drinfeld presentation.} \cite{FHHSY,FT_Fock} gives rise to an interesting integrable model. The study of this model was initiated in \cite{FJMM_FTM,FJMM_BA} using methods of representation theory and a new approach in terms of the related shuffle algebra circumventing the classical approach which is typically applied to the Heisenberg type models associated to the quantum affine Lie algebras $\sln$. 

A standard tool in the diagonalization problem of Heisenberg type models is the $R$-matrix. In the case of fundamental representations of $\sln$ the $R$-matrices are known explicitly and have a simple form. These $R$-matrices are finite dimensional and have a few non-zero entries. On the other hand the $R$-matrix of $\qg$ in the Fock representation is an infinite dimensional matrix with infinitely many non-zero entries which are rational functions in two parameters $q$ and $t$ as well as in the spectral parameter $u$. The calculation of this $R$-matrix turns out to be a difficult problem. From the point of view of algebraic geometry the problem of toroidal $R$-matrices was considered in particular in \cite{MOk,Smirn16,Smirn20,Neg15,NegMod}. The shuffle algebra point of view was taken in \cite{NegShuf} and the algebraic approach which uses the coproduct equation was considered in \cite{Fukud,Awata,AwataRTT}, where the authors compute first few coefficients of the expansion of the $R$-matrix in the spectral parameter.

In this paper we consider the algebraic approach to the computation of the Fock $R$-matrix of $\gt$. We consider $R$ as the solution of the set of equations arising from the coproduct relations. We manage to reduce these equations to a single equation which is simpler and more transparent than the original set of equations. We make a connection to the Macdonald theory of symmetric polynomials and derive an equation for the matrix elements of $R$ which is expressed through fundamental operators of the Macdonald theory: the Macdonald difference operator and the Cauchy kernel. As the first application of these results we derive a recursive formula which allows one to compute the matrix elements of $R$ written in the basis of Macdonald polynomials. In addition, we speculate that our strategy,  i.e. finding a reduced equation for the $R$-matrix from the coproduct relations, its reformulation in terms of symmetric functions and the recursive formula, can be extended by analogy to other quantum groups with Fock representations.

\subsection{The FJMM model}
One of the main motivations for this work is the study of the integrable model on the tensor product of Fock spaces of $\qg$ \cite{FJMM_FTM,FJMM_BA,FJMM_IM}, which we call the FJMM (Feigin--Jimbo--Miwa--Mukhin) model. The FJMM model describes oscillators on a one dimensional lattice with next neighbour interaction. The model depends on two parameters $q,t\in \bbC^*$. The Hamiltonian is given as a sum of two-site Hamiltonians which are expressed through the Heisenberg algebra which is generated by  $\{h_{-r},h_r\}_{r=1}^\infty$ satisfying
\ba
[h_r,h_s]= \d_{r,-s} \frac{1}{r}\frac{1}{(t^{-r}-q^r)(1-q^{-r}t^{r})}.
\ea
The operators $h_r$ and $h_s$ commute if $\text{sign}(r)=\text{sign}(s)$ and thus for any partition $\m=(\m_1,\dots,\m_n)$ we can introduce 
\ba
h^*_\m:= h_{-\m_1}\dots h_{-\m_n},
\qquad
h_\m:= h_{\m_1}\dots h_{\m_n}. 
\ea
The Hamiltonian acting on the tensor product of $L$ copies of the Fock spaces is
\begin{align*}
H=\sum_{i=1}^L H_{i,i+1},
\qquad 
H_{i,i+1}=
J \sum_{\a,\b,\g,\d}
 h^*_{\a} h_{\b} \otimes h^*_{\g} h_{\d},
\end{align*}
where the sum runs over all partitions $\a,\b,\g$ and $\d$, such that $|\a|+|\g|=|\b|+|\d|$ and $J=J(q,t)$ is some parameter. One can then, for example, impose periodic boundary conditions by demanding $H_{L,L+1}=H_{1,L}$. With the knowledge of the $R$-matrix $R(u)$ the diagonalization problem of $H$ is replaced with the diagonalization of the transfer matrix $T(u)$
\ba
\label{Tmx}
T(u;u_1,\dots, u_L)=
\text{Tr}_{V_0(u)}
\s 
R_{0,1}(u_L/u)\dots R_{0,L}(u_1/u),
\ea
where $\s$ is a twist operator which regularizes the trace.  As usual the Hamiltonian $H$ is computed by taking the logarithmic derivative of the transfer matrix $T(u)$
\ba
H=
T(0)^{-1}\frac{\dd}{\dd x} T(u)|_{u=0} +\text{const}.
\ea
A standard diagonalization method for the XXZ transfer matrix is the algebraic Bethe Ansatz \cite{Fad}. This method relies on the explicit knowledge of $R(u)$. 

\subsection{Computation of the Fock $R$-matrix using the coproduct}
Many quantum toroidal and quantum affine algebras can be realised as a quantum double using the notion of pairing, a certain bilinear form on the generators, see \cite{Drinfeld}. Let $\MU$ denote one such quantum group. If elements $a\in\MU$ represent an orthonormal basis $\MB$ with the paired dual elements $a^*$ then one immediately has a formula for the {\it universal matrix} $\MR\in \MU\otimes \MU$
\begin{align}
\MR=\sum_{a\in \MB} a\otimes a^*.
\end{align}
This matrix $\MR$ satisfies the equations
\begin{gather}
\label{eqR}
\MR\Delta(g)=\Delta'(g) \MR, \\
\label{eqdR}
(\Delta\otimes \id)\MR=
\MR_{1,3}\MR_{2,3}, \qquad
(\id \otimes \Delta)\MR=
\MR_{1,3}\MR_{1,2},
\end{gather}
where $\D$ is the coproduct and $\D'$ is the opposite coproduct of $\MU$.  Equation~\eqref{eqR} is the coproduct equation which holds for all $g\in \MU$; the second and the third equations take place in  $\MU^{\otimes 3}=\MU\otimes \MU\otimes \MU$. The element $\id$ is the identity and $\MR_{1,2}=\MR\otimes \id$, $\MR_{2,3}=\id\otimes \MR$ and $\MR_{1,3}=(\id\otimes \t) \MR_{1,2}$ ($\t$ being the transposition $\t(a\otimes b)=b\otimes a$). It follows from (\ref{eqR}) and  (\ref{eqdR}) that such an element $\MR$ satisfies  the universal Yang--Baxter equation
\begin{align}
\label{eqYB}
\MR_{1,2}\MR_{1,3}\MR_{2,3}
=
\MR_{2,3}\MR_{1,3}\MR_{1,2}.
\end{align}

If $\MR$ is specialised to a particular representation then it defines an integrable model whose $R$-matrix satisfies (\ref{eqYB}). In the case $\MU=\gt$, the quantum toroidal algebra, the pairing does not correspond to the orthonormal basis of the standard Drinfeld presentation of elements\footnote{One can write the universal $R$-matrix in the elliptic Hall formulation of $\gt$, see \cite{NegShuf}} of $\gt$ (see Section \ref{sec:algebra}). This problem arises due to the difficulty of finding a Poincar\'e--Birkhoff--Witt basis for $\gt$ using the Drinfeld generators.  This problem disappears after we take a suitable representation of $\gt$. In this case (\ref{eqR}) becomes the defining equation for the matrix $R$ in the chosen representation. Using explicit formulas \cite{FHHSY} of the action of $\gt$ on the Fock space $V$ we can write a linear equation for $R$
\begin{align}
\label{eqRF}
R\Delta(g)=\Delta'(g) R,
\end{align}
where $\Delta(g)$ is the Fock image of the coproduct on elements $g\in \gt$. 

We expect that our treatment of (\ref{eqRF}) can also be applied for finding $R$-matrices when the algebra $\gt$ is replaced with a different quantum affine and toroidal algebra in its Fock representation. The reason lies in the similarities of structures of such algebras and their coproducts. Their presentation is often given by ``Cartan currents'' $\psi^{\pm}(z)$, ``raising'' currents $e(z)$ and ``lowering'' currents $f(z)$. The currents  $\psi^{\pm}(z)$ are expressed through Heisenberg operators $h_r,h_{-r}$, $r>0$. The currents $e(z)$ and $f(z)$ are expressed as  vertex operators which have the form of exponentials in $h_{\pm r}$. The coproduct equation (\ref{eqRF}) then needs to be satisfied for $g\in\{h_{\pm r},e(z),f(z)\}$.  For higher rank algebras all operators acquire additional indices. 

We suggest a procedure for the reduction of (\ref{eqRF}) to a single equation\footnote{In the case of higher rank algebras one should expect a set of equations.} in which the elements of the algebra in their vertex operator representation act on a vector $\DR\in V\otimes V$ whose components coincide with the elements of the $R$-matrix.
\begin{enumerate}[label={\arabic*}.]
\item In the first step we factorise the matrix $R$,
\begin{align*}
R=K \Rb,
\end{align*}
where the first factor $K$ is simple and can be computed with the knowledge of the pairing of the Drinfeld quantum double construction. The second part $\Rb$ is the difficult part, it depends on the currents $e(z)$ and $f(z)$ given by vertex operators.  

\item In the next step we need to compute a ``twisted'' opposite coproduct $\widetilde{\Delta}(g)$ using the definition
\begin{align*}
\widetilde{\Delta}(g):=K^{-1}\Delta'(g)K.
\end{align*}
With this the problem is rephrased in terms of the unknown part $\Rb$
\begin{align}
\label{eqRb}
\Rb\Delta(g)=\widetilde{\Delta}(g) \Rb.
\end{align}
Consider this equation for $g=h_r$. In the case of the toroidal quantum algebra $\gt$ we find $\widetilde{\Delta}(h_r)=\Delta(h_r)$, hence
\begin{align}
\label{eqRht}
[\Rb,\Delta(h_r)]=0.
\end{align}
Since $\Delta$ is a homomorphism, this means that the algebra generated by $b_{\pm r}:=\Delta(h_{\pm r})$ forms a new Heisenberg algebra. In the tensor product of two modules we have two canonical commuting families of Heisenberg operators generated by $1\otimes h_r$ and $h_r\otimes 1$. The operators $b_{\pm r}$ are given by a linear combination of these elements. We can find another family of Heisenberg operators $c_{\pm r}$ as a different linear combination, such that $[b_r,c_s]=0$ for all $r,s\neq 0$. It follows that $\Rb$ is expressed in terms of only one family of operators $c_{\pm r}$. 

\item In the third step we rewrite (\ref{eqRb}) with $g=e(z),f(z)$ in terms of $c_{\pm r}$. Changing the basis from $\{1\otimes h_{\pm r},h_{\pm r}\otimes 1\}$ to $\{b_{\pm r},c_{\pm r}\}$ leads to an equation
\begin{align}
\label{VOR}
A(z)\Rb=\Rb B(z),
\end{align}
where two new operators $A(z)$ and $B(z)$ are expressed via $c_{\pm r}$ only.
\item In the last step we view (\ref{VOR}) as an equation in the vector space of the universal enveloping algebra of the Heisenberg algebra and consider $A(z)$ and $B(z)$ as operators acting on $\Rb$. We then use an isomorphism between the vector space of the enveloping Heisenberg algebra and the tensored Fock space $V\otimes V$. This isomorphism replaces $\Rb$ with the corresponding vector $\DR\in V \otimes V$, so that we can rewrite (\ref{VOR}) in vector form
\begin{align}
\label{VORvec}
\mathcal{A}(z)\DR=\mathcal{B}(z) \DR.
\end{align}
In this equation $\mathcal{A}(z)$ and $\mathcal{B}(z)$ are operators which are derived from $A(z)$ and $B(z)$ using the isomorphism. The operators $\mathcal{A}(z)$ and $\mathcal{B}(z)$  can be rewritten again in terms of vertex operator representation of the elements $e(z)$, $f(z)$  and $\psi^{\pm}(z)$. We subsequently take the constant term in $z$ of (\ref{VORvec}) and expand $\DR$ in the eigenbasis of the zero modes of $e(z)$ and $f(z)$. This leads to an equation on the components of $\DR$. In the case of the toroidal algebra $\gt$ this equation represents a recurrence relation from which we can compute all components of $\DR$ starting from the initial condition.
 
\end{enumerate}

\subsection{The Fock $R$-matrix of $\gt$}
In this section we give an informal summary of how the above strategy works for the Fock $R$-matrix of $\gt$. Before we do so we make a remark about the choice of parameters. The toroidal algebra $\gt$ has connections, in particular, to the Macdonald theory of symmetric functions and quantum affine algebras. We have two independent complex parameters $q$ and $t$ but in addition we will encounter two more sets: 
\ba
\{q_1,q_2,q_3\}, \qquad 
\{\qq,\tt\},
\ea
which are related as follows
\ba
&\tt=q_1=\frac{1}{q t}, \qquad \qq=q_3^{-1} = \frac{q}{t},
\qquad q_2 =q^2.
\ea
For the definition of $\gt$ we follow \cite{FJMM_BA} and use $\{q_1,q_2,q_3\}$ and $q,t$. The parameters $\{\qq,\tt\}$ are matched with the two parameters of the Macdonald theory \cite{MacdBook}, s.t. $P_{\l}(x;\qq,\tt)$ denotes a Macdonald polynomial.

Let us now turn to the four steps of the calculation. The $R$-matrix $R(u)\in \text{End}(V\otimes V)$ depends on the ratio of two spectral parameters $u=u_2/u_1$, where $u_1,u_2$ are attached to the two Fock representations $V$ in the tensor product. We recall the Heisenberg algebra commutation relations
\begin{align*}
[h_r,h_s]= \d_{r,-s} \frac{q^{r}-q^{-r}}{r \k_r}, 
\qquad \k_r=(1-q_1^r)(1-q_2^r)(1-q_3^r).
\end{align*}
The currents $e(z)$ and $f(z)$ in the Fock representation of $\gt$ can be written as vertex operators \cite{FHHSY}
\begin{align}
\label{eq:evointro}
&
e(z)=\frac{1-q_2}{\k_1} u ~
\exp\left(\sum_{r=1}^{\infty}\frac{\k_r}{1-q_2^{r}} h_{- r}z^{ r}\right)
\exp\left(\sum_{r=1}^{\infty}\frac{q^r \k_r}{1-q_2^{r}} h_{ r}z^{- r}\right),\\
\label{eq:fvointro}
&
f(z)=-\frac{1-q_2}{q_2 \k_1}u^{-1} 
\exp\left(-\sum_{r=1}^{\infty}\frac{q^r \k_r}{1-q_2^{r}} h_{- r}z^{ r}\right)
\exp\left(-\sum_{r=1}^{\infty}\frac{q^{2r} \k_r}{1-q_2^{r}} h_{ r}z^{- r}\right).
\end{align}
In the first step we factorise the matrix $R(u)$ as in \cite{FJMM_BA}
\begin{align}
\label{Rmatrixfactor}
&R(u)=
\exp\left( \sum_{r\geq 1}r \k_r h_{r} \otimes h_{-r} \right)
q^{-d_1-d_2}
\Rb(u).
\end{align}
In the second step we make use of the commutativity (\ref{eqRht}), from which it follows that the matrix $\Rb(u)$ can be expressed in a new Heisenberg basis 
\begin{align}
\label{eq:Rbintro}
\Rb(u)
     =
\sum_{\m,\n}
\Rb_{\m,\n}(u)
c^*_{\m} c_\n,
\end{align}
where the sum runs over all partitions $\m,\n$ and $\Rb_{\m,\n}(u)$ are non-zero for partitions of the same weight $|\m|=|\n|$. The new Heisenberg operators $c_r$ satisfy
\begin{align*}
 [c_r, c_{-s}]
=
\d_{r,s}
 \frac{(q^r+q^{-r})(1-q_1^r)(1-q_3^r)}{r}, \qquad r,s>0.
\end{align*}
In the third step we rewrite the coproduct relations (\ref{eqRb}) with $g=e(z),f(z)$ in terms of this new Heisenberg algebra. The two resulting  equations for $\D(e(z))$ and $\D(f(z))$ coincide and can be conveniently expressed via two new vertex operators $\varphi^{\pm}(z)$
\begin{align*}
\varphi^{\pm}(z):=
\exp\left(\mp \sum_{r=1}^{\infty}\frac{1}{q^r+q^{-r}}c_{-r}z^{ r}\right)
\exp\left(\pm \sum_{r=1}^{\infty}\frac{1}{q^r+q^{-r}}c_r q^r z^{-r}\right).
\end{align*}
Using the operators $\varphi^{\pm}(z)$ and the coproduct we derive (Proposition \ref{propRFE}, see also \cite{Fukud})
\begin{align}
\label{eq:VOR}
[\varphi^+(z),\Rb(u)]
=u^{-1}\left( \Rb(u) \varphi^-(z q^{-2}) - \varphi^-(z q^2)\Rb(u) \right).
\end{align}
Equation \eqref{eq:VOR} is the explicit form of (\ref{VOR}) in the case of the Fock representation of $\gt$, and the central starting point to compute $\Rb(u)$. 

In the fourth step (\ref{eq:VOR}) is turned into a vector equation in which we think of $\vphi^{\pm}(z)$ in (\ref{eq:VOR}) as operators acting on $\Rb(u)$. The Heisenberg basis elements $c^*_{\m} c_\n$ in (\ref{eq:Rbintro}) are mapped to basis elements of $V\otimes V$. As a result we can rewrite the left and right actions of $\vphi^{\pm}(z)$ as left actions by vertex operators on vectors in $V\otimes V$. The $R$-matrix $\Rb$ is replaced with a vector $\DR\in V\otimes V$ and we get an equation of the form (\ref{VORvec}) (Proposition \ref{prop:copef}). Taking the coefficients of $z^0$ in this equation leads to (Theorem \ref{eq:thmzero})
\begin{align}
\label{eq:copredintro}
u  \left(
\frac{\k_1 q^2}{1-q_2} ~ 1\otimes  \fd_0 + 1\otimes1\right)\DR(u) 
= \left(
 1\otimes1 - \frac{\k_1}{1-q_2}  \sum_{j\geq 0} q^{j} \ed_{-j}\otimes  \psid^-_{-j}
 \right)\DR(u),
\end{align}
where $\fd_0$ and $\ed_{-j}$ are the modes of the same vertex operators $e(z), f(z)$ in (\ref{eq:evointro}) and (\ref{eq:fvointro}) but with $u=1$, and $\psid^{-}_{-j}$ are the modes of the Cartan current $\psi^-(z)$.  Equation (\ref{eq:copredintro}) represents a single equation that determines the unknown coefficients of $\Rb(u)$ by computing the vector $\DR(u)$. 

To achieve the latter, we use the isomorphism between the Fock space and the ring of symmetric functions $\Lambda$. In this form we are able to rewrite (\ref{eq:copredintro}) in terms of the Macdonald difference operator $E$, its dual $\Eb$ and the Cauchy kernel $\Pi(x,y)$ (see \cite{MacdBook}), where $x=(x_1,x_2,\dots)$ and $y=(y_1,y_2,\dots)$ are the two copies of alphabets of the two rings $\Lambda$. The operators $E$ and $\Eb$ are the images under the isomorphism of $\ed_0$ and $\fd_0$.  We also require a modified Cauchy kernel $\Pit(x,y)$
\ba
\Pi(x,y)=
\exp\left( \sum_{r\geq 1}\frac{(1-\tt^r)}{r(1-\qq^{r})} p_r(x)p_r(y)  \right), 
\qquad \Pit(x,y):=\Pi(q x,\qq y)^{-1} \Pi(q x,\tt y),
\ea
where $p_r(x)$ and $p_r(y)$ are the power sum symmetric functions in the alphabets $x$ and $y$. By applying the isomorphism to $\DR(u)$ we get a symmetric function version of this vector which we denote by $R(x,y;u)$. 
With this we rewrite (\ref{eq:copredintro}) as (Theorem \ref{thm:DPR})
\begin{align}
\label{eq:intro}
E_x\Pit(x,y)  R(x,y;u)
=u~
\Pit(x,y) \bar{E}_y   R(x,y;u).
\end{align}
From \eqref{eq:intro} it is clear that the relevant basis for $R(x,y;u)$ is the set of Macdonald polynomials since the operators $E_x, \Eb_y$ act diagonally on them. We therefore define new expansion coefficients $L_{\a,\b}(u)$
\begin{align}
\label{Rmac}
R(x,y;u)
     =
\sum_{\a,\b}
L_{\a,\b}(u)
P_\a(x;\qq,\tt)P_\b(y;\qq,\tt).
\end{align}
The coefficients $L_{\alpha,\beta}(u)$ are non-vanishing for $|\alpha|=|\beta|=w$, $w\geq0$. Using the initial condition $L_{\pzero,\pzero}(u)=1$ and equation (\ref{eq:intro}) we determine these coefficients recursively with a formula of the form
\begin{align}
\label{eq:Rbr}
L_{\alpha,\beta}(u)=\sum_{\m\prec\alpha,\n\prec\beta} L_{\alpha/\m,\beta/\n}(u) L_{\m,\n}(u),
\end{align}
where the summation runs over all $\m$ and $\n$ such that $\a/\m$ and $\b/\n$ represent a pair of skew Young diagrams each having at least one box. A formula for the skew functions $L_{\alpha/\m,\beta/\n}(u)$ is given in Proposition \ref{prop:Lbranch}. This formula provides further insights and can also be employed to compute $L_{\alpha,\beta}(u)$ for partitions $\a,\b$ of small weights (up to about $w=6$).

\subsection{Evaluation of the $R$-matrix and relation to the six vertex model}
After determining the coefficients $\Rb_{\m,\n}$, using (\ref{Rmac}) and the transition coefficients from the Macdonald polynomials to the power sums, we can explicitly compute matrix elements of the full $R$-matrix corresponding to partitions of small weights. We insert $\Rb$ into (\ref{Rmatrixfactor}), rewrite it using the Heisenberg algebra $h_{\pm r}$, and then sandwich the resulting operator $R(u)$ between two states of the tensor product in the standard Fock realization. After this evaluation we find
\begin{align}
\label{Rmatrix}
R(u)=\left[
\begin{array}{cccccc}
 1 & 0 & 0 & 0 & 0 &\dots\\
 0 & \frac{q(1-u)}{1-q^2 u} & \frac{(1-q^2 )u}{1-q^2 u} & 0 & 0 &\dots\\
 0 & \frac{1-q^2}{1-q^2 u} & \frac{q(1-u)}{1-q^2 u}  & 0 & 0 & \dots\\
 0 & 0 & 0 & R^{(2)}(u)  & 0 & \dots\\
0 &  0 & 0 & 0 & R^{(3)}(u) & \dots\\
\vdots & \vdots  & \vdots & \vdots  & \vdots & \ddots
\end{array}
\right].
\end{align}
This $R$-matrix has an expected block structure.  The matrix elements, which we denote by $[R(u)]_{\a,\b}^{\g,\d}$, are rational functions in $u,q,t$, they give Boltzmann weights of a vertex model whose ``in'' and ``out'' states are labelled by Young diagrams $\a,\b$ and $\g,\d$, respectively. The non-zero elements correspond to $|\a|+|\b|=|\g|+|\d|$ which reflects a typical conservation law in vertex models. 

There are infinitely many blocks $R^{(w)}$, and a single block contains (all non-vanishing) elements $[R(u)]_{\a,\b}^{\g,\d}$ with $\a,\b,\g,\d$ satisfying $|\a|+|\b|=|\g|+|\d|=w$. The sizes of the blocks $R^{(w)}$ with $w=2,3,4,5$ are $n\times n$ with $n=5,10,20,36$, respectively. The $5\times 5$ block $R^{(2)}$ is already very large which indicates that the language of (Macdonald) symmetric functions (\ref{Rmac}) is more suitable to give explicit expressions for such $R$-matrices. 

We note that the blocks $R^{(0)}$ and  $R^{(1)}$, which are given explicitly in (\ref{Rmatrix}), are equal to the first and the second blocks of the $R$-matrix of the six vertex model (see \cite{Bax})\footnote{The six vertex $R$-matrix (\ref{Rsix}) is written in one of the standard conventions which justifies our convention on the use of the parameter $q$.}
\begin{align}
\label{Rsix}
R_{6V}(u)=\left[
\begin{array}{cccc}
 1 & 0 & 0 & 0 \\
 0 & \frac{q(1-u)}{1-q^2 u} & \frac{(1-q^2 )u}{1-q^2 u} & 0 \\
 0 & \frac{1-q^2}{1-q^2 u} & \frac{q(1-u)}{1-q^2 u}  & 0 \\
 0 & 0 & 0 & 1 \\
\end{array}
\right].
\end{align}
The last block in the six vertex $R$-matrix, given by a single entry $1$, can be matched with the matrix element $[R(u)]_{(1),(1)}^{(1),(1)}$ which belongs to the block $R^{(2)}$ in (\ref{Rmatrix}). This matrix element has the form $1+t ~c(u;q,t)$ with $c(u;q,t)$ regular as $t\rightarrow 0$. Therefore we recover $1$ in the limit $t=0$.

\subsection{Outline}
The outline of this paper is the following. We give the definition of the algebra $\gt$ in Section~\ref{sec:algebra}. In Section~\ref{sec:Fock} we recall the Fock module. In Section~\ref{sec:coprod} we reduce the coproduct equations to an equation involving vertex operators and then rewrite this equation in the vector form. In Section~\ref{sec:MacdonaldR} we turn to the Macdonald theory and derive an equation for the symmetric function $R(x,y;u)$ and then obtain a recursive formula for the coefficients of $R(x,y;u)$ in the Macdonald basis.


\section{The quantum toroidal algebra $\gt=\qg$}
\label{sec:algebra}
In this section we fix the definitions of the algebra $\gt$. We work in the Drinfeld presentation and our definitions coincide with those chosen in \cite{FJMM_BA}. The algebra $\gt$ depends on two parameters $q,t\in \bbC$. For the defining relations it is convenient to use a related set of parameters $q_1, q_2$ and $q_3$, satisfying $q_1 q_2 q_3 =1$, and
\begin{align}
\label{eq:qt_tor}
q_1=\frac{1}{q t}, \qquad q_2 =q^2, \qquad q_3 = \frac{t}{q}.
\end{align}
With these parameters we define a function $g(z,w)$ and coefficients $\k_r$
\begin{align}
\label{eq:g}
&g(z,w):=(z-q_1 w)(z-q_2 w)(z-q_3 w),\\
\label{eq:kappa}
&\k_r:=(1-q_1^r)(1-q_2^r)(1-q_3^r),
\end{align}
which are used to write the defining relations of the algebra.


\subsection{Generators and relations}
The algebra $\gt$ is generated by the currents 
\begin{align}\label{currents}
e(z), \quad f(z), \quad \psi^{\pm}(z),
\end{align}
two central elements $C=q^c$ and $\Cp=q^{\cp}$ and two grading operators $D=q^d$ and $\Dp=q^{d^{\perp}}$. These operators satisfy the following relations
\begin{gather}
[C,\Cp]=0,\qquad [D,\Dp]=0, \nonumber \\
\label{eq:Dgrading}
D e(z)= e(qz)D, \qquad 
D f(z)= f(qz)D, \qquad 
D \psi^{\pm}(z)= \psi^{\pm}(qz)D,
\\
\Dp e(z)= q e(z)\Dp, \qquad 
\Dp f(z)= q^{-1}f(z)\Dp, \qquad 
\Dp \psi^{\pm}(z)= \psi^{\pm}(z)\Dp, \nonumber
\\
\label{eq:psippsim}
\psi^{\pm}(z)\psi^{\pm}(w)= \psi^{\pm}(w)\psi^{\pm}(z),
\\
\label{eq:psipm}
\frac{g(q^{-c}z,w)}{g(q^c z,w)}
\psi^{-}(z)\psi^{+}(w)
= 
\frac{g(w,q^{-c}z)}{g(w,q^c z)}
\psi^{+}(w)\psi^{-}(z),
\\
\label{eq:ewez}
g(w,z)e(w)e(z)
+g(z,w)e(z)e(w)=0,
\\
\label{eq:fwfz} 
g(w,z)f(z)f(w)
+g(z,w)f(w)f(z)=0, 
\\
\label{eq:ewpsiz}
g(w,z)e(w)\psi^{\pm} (q^{\frac{1}{2}(-1\mp1)c}z)
+g(z,w)\psi^{\pm}(q^{\frac{1}{2}(-1\mp1)c}z)e(w)=0, 
\\ 
\label{eq:fwpsiz}
g(w,z)\psi^{\pm}(q^{\frac{1}{2}(-1\pm1)c}z)f(w)
+g(z,w)f(w)\psi^{\pm}(q^{\frac{1}{2}(-1\pm1)c}z)=0, 
\\
\label{eq:ezfw}
[e(z),f(w)]=
\frac{1}{\k_1}\left(\d(q^c \frac{w}{z})\psi^+(w)-\d(q^c\frac{z}{w})\psi^-(z)\right),
\end{gather}
supplemented by the cubic relations
\begin{align}
\underset{z_1,z_2,z_3}{\text{Sym}}z_2 z_3^{-1} [e(z_1),[e(z_2),e(z_3)]]=0= \underset{z_1,z_2,z_3}{\text{Sym}}z_2 z_3^{-1} [f(z_1),[f(z_2),f(z_3)]].
\label{eq:fff}
\end{align}

\subsubsection{Mode expansions}
The currents (\ref{currents}) have the following mode expansions
\begin{align}
\label{modes}
e(z)=\sum_{n\in \bbZ} e_n z^{-n},
\qquad 
f(z)=\sum_{n\in \bbZ} f_n z^{-n},
\qquad 
\psi^{\pm}(z)=\sum_{j=0}^{\infty} \psi^{\pm}_{\pm j} z^{\mp j}.
\end{align}
Along with the modes $\psi^{\pm}_{j}$ one also uses the modes $h_r$, $h_{-r}$, with $r>0$, defined by
\begin{align}
\label{psiexp}
\psi^{\pm}(z)=q^{\mp \cp} \exp\left(\sum_{r=1}^{\infty} \k_r h_{\pm r} z^{\mp r}\right). 
\end{align}
After expanding the exponential we get a relation between $\psi^{\pm}_{\pm j}$ and $h_{\pm r}$
\begin{align}
\label{psih}
&
\psi_{\pm j}^{\pm}
=
  q^{\mp\cp}
\sum_{m=0}^{j}
 \frac{1}{m!}
\sum_{\substack{r_1,\dots,r_m >0\\r_1+\cdots + r_m=j}}\k_{r_1}\dots \k_{r_m} h_{\pm r_1}\dots h_{\pm r_m},
\end{align} 
Substituting the mode expansions of $e(z)$, $f(z)$ and $\psi^{\pm}(z)$ in the form (\ref{psiexp}) into relations (\ref{eq:psippsim})--(\ref{eq:ezfw}) we find
\begin{align}
\label{eq:hh}
\begin{split}
[h_r,h_s] &= \d_{r,-s} \frac{q^{c r}-q^{-c r }}{r \k_r},\\
[h_r,e_n] &= -\frac{1}{r}e_{n+r}q^{\frac{1}{2}c (-r-|r|) },\\
[h_r,f_n] &=  \frac{1}{r}f_{n+r}q^{\frac{1}{2}c (-r+|r|)},\\
[e_k,f_{-n}] &= 
\frac{1}{\k_1}
\left(\d_{k\geq n} q^{c k}  \psi^{+}_{k-n}
-
\d_{k\leq n}q^{-c n} \psi^{-}_{k-n}\right),
\end{split}
\end{align}
where the Kronecker delta symbol with a single argument is  defined by $\d_{\text{True}}=1$ and $\d_{\text{False}}=0$. From these relations it is clear that the elements $e_0$, $f_0$, $h_1$  and $h_{-1}$ can be used to generate all other modes in the algebra. Therefore they form a generating set of $\gt$ together with the degree operators and the central elements.

\subsubsection{Grading and decomposition}
The defining relations (\ref{eq:Dgrading}) equip the algebra $\gt$ with a $\bbZ^2$ grading. For each element $g\in \gt$ the operators $\Dp$ and $D$ define the {\it principal} degree and the {\it homogeneous} degree, respectively
\begin{equation}
\begin{split}
\label{degdef}
&\deg(g)=(\pdeg(g),\hdeg{g}),
\\
&q^{d^{\perp}} g q^{-d^{\perp}} = q^{\pdeg(g)} g, 
\qquad 
q^d g q^{-d}= q^{-\hdeg(g)} g.
\end{split}
\end{equation}
With (\ref{degdef}), relations (\ref{eq:Dgrading}) lead to
\begin{equation}
\begin{split}
\label{degsefh}
&\deg(e_n)=(1,n),\qquad
\deg(f_n)=(-1,n),\qquad
\deg(h_r)=(0,r),\\
&\deg(g)=(0,0),\qquad g\in(C,\Cp,D,\Dp).
\end{split}
\end{equation}

As we can see from (\ref{eq:hh}) the generators $h_{\pm r}$, $r>0$, form  a Heisenberg subalgebra. Using the generators $e_n,f_n$ ($n\in \mathbb{Z}$) and $h_{\pm r}$ ($r>0$), we define two natural subalgebras 
\begin{equation}
\label{eq:gsub}
\begin{split}
g_{\geq} &:= \langle e_n, h_r , C, C^{\perp}, D, D^{\perp}\rangle,
\qquad
n\in \mathbb{Z},\ r>0,
\\
g_{\leq} &:= \langle f_n, h_r , C, C^{\perp}, D, D^{\perp}\rangle,
\qquad n\in \mathbb{Z},\ r>0.
\end{split}
\end{equation}
We will also consider the subalgebras
\begin{equation}
\label{eq:gsubstrict}
\begin{split}
g_{>} := \langle e_n\rangle,\ n\in \mathbb{Z},
\qquad g_{<} := \langle f_n \rangle,
\ n\in \mathbb{Z}.
\end{split}
\end{equation}


\subsection{Miki's automorphism}
\label{ssec:th}
In $\gt$ there exists an order four automorphism \cite{Miki}, denoted $\th$. Its action on the generators is given by
\begin{gather}
\th(e_0) = h_{-1}, \qquad 
\th(f_0) = h_1, \qquad 
\th(h_1) = e_0, \qquad 
\th(h_{-1}) = f_0, 
\\
\th(\Cp) = C, \qquad 
\th(C) = (\Cp)^{-1}, \qquad 
\th(\Dp) = D, \qquad 
\th(D) = (\Dp)^{-1}.
\end{gather}


\subsection{The coproduct}
The algebra $\gt$ is a Hopf algebra with the coproduct given by
\begin{align}
\label{eq:De}
\begin{split}
\D(e_k) &= \sum_{j=0}^{\infty} e_{k-j}\otimes q^{c k }  \psi_j^{+} + 1\otimes e_k, \\
\D(f_k) &=  f_k\otimes 1 + \sum_{j=0}^{\infty} q^{c k } \psi_{-j}^{-} \otimes  f_{k+j}, \\
\D(h_r) &=  h_r\otimes 1 + q^{-c r }  \otimes h_r , \\
\D(h_{-r}) &= h_{-r}\otimes q^{c r }  +  1 \otimes h_{-r},
\\
\D(g) &= g\otimes g, \qquad g\in(C,\Cp,D,\Dp).
\end{split}
\end{align}
By summing over the indices $k$ and $r$ in (\ref{eq:De}) one obtains the coproduct for the currents $e(z),f(z)$ and $\psi^{\pm}(z)$. Let $\t$ be the transposition operation $\t(a\otimes b)=b\otimes a$ for any $a,b\in \gt$. Then the opposite coproduct $\D'$ is defined by 
\ba
\D'(g)=\t \D(g), \qquad g\in \gt.
\ea

\subsection{The universal $R$-matrix}
\label{sec:Rfact}
Let us discuss the universal $R$-matrix in the Drinfeld presentation (for the discussion in the elliptic Hall setting see for example \cite{NegShuf}).
The universal $R$-matrix, denoted $\MR$, is an element in a completion of $\gt_{\geq}\otimes \gt_{\leq}$, with $\gt_{\geq}$ and $\gt_{\leq}$ defined in \eqref{eq:gsub}. The $R$-matrix satisfies three relations with the coproduct
\begin{align}
\label{deltaUR}
&\MR\D(g)=\D'(g) \MR, \qquad \forall g\in \gt,
\\
\label{URR1}
&(\D\otimes 1)\MR=
\MR_{1,3}\MR_{2,3}, \\
\label{URR2}
&(1 \otimes \D)\MR=
\MR_{1,3}\MR_{1,2}.
\end{align}
The element $\MR$ then satisfies the Yang--Baxter equation
\ba
\MR_{1,2}\MR_{1,3}\MR_{2,3}
=
\MR_{2,3}\MR_{1,3}\MR_{1,2}.
\ea
Using the pairing of the quantum double of $\gt$ (see e.g. \cite{FJMM_BA} for details) one can deduce the factorized form for the universal $R$-matrix 
\begin{align}
\label{KR}
\MR=\MK\MRb,
\end{align}
where the first factor $\MK$ reads
\begin{align}
\label{MK}
&\MK=\exp\left( \sum_{r\geq 1}r \k_r h_{r} \otimes h_{-r} \right)
q^{-c\otimes d-d\otimes c -\cp\otimes d^{\perp}-d^{\perp}\otimes \cp },
\end{align}
and the second factor can be written only up to the first term in the homogeneous degree
\begin{align}
\label{MR2}
&\MRb=1\otimes 1+ \k_1 \sum_{n\in \bbZ} e_{n} \otimes f_{-n} +\dots.
\end{align}
The terms of the homogeneous degree higher than $1$ are not known since the pairing is not diagonal. We note that all terms in (\ref{MR2}) must be neutral in the principal degrees.  

The linear equation (\ref{deltaUR}) in $\MR$ will be referred to as the coproduct relation. In order to compute $\MR$ by solving (\ref{deltaUR}) it is necessary to have a Poincar\'e--Birkhoff--Witt (PBW) basis in $\gt$ in the chosen presentation. Such a basis will allow one to express a product of two elements in $\gt$ as a finite linear sum of the basis elements. This will lead to a system of equations whose unknowns are the basis expansion coefficients of $\MR$. Solving these equations would provide an explicit form of $\MR$. 

It is however not clear how to build the PBW basis for $\gt$ in the Drinfeld presentation. The issue arises when one tries to build a basis for $\gt_{>}$ and $\gt_{<}$ defined in \eqref{eq:gsubstrict}. A direct algebraic construction of this basis in the case of $\gt_{>}$ could be carried out as follows. We need to identify the basis elements in each subspace with fixed homogeneous degree. For the homogeneous degree equal to $1$ we simply choose all elements $e_i$. For the homogeneous degree equal to $2$ one typically takes the basis to consist of elements $e_i e_j$, for $i,j \in \bbZ$ and requires $i\leq j$ (or $i\geq j$).  If such a basis is present then, using the commutation relations in $\gt$, one should be able to write all elements $e_k e_l$ with $k>l$ as finite linear sums in  $e_i e_j$ with $i\leq j$. In other words, we need to have an ordering relation of the form
\ba
e_{k}e_l = \sum_{\substack{i^*\leq i\leq j\leq j^*,\\ i+j=k+l}}c_{i,j} e_{i} e_{j}, \qquad k>l,
\ea
where $i^*$ and $j^*$ are some finite integers and $c_{i,j}$ are some coefficients which may also depend on $k$ and $l$. This basis fails because such quadratic identities do not hold in $\gt$. To the authors' knowledge, constructions of PBW bases for $\gt$ are absent in the literature. For certain homomorphisms, like the Fock representation discussed in Section \ref{sec:Fock},  there are natural candidates for PBW bases and the ordering relations can be written explicitly.


\section{The Fock module}
\label{sec:Fock}
We use the same notations as in  \cite{FJMM_BA}. The action of $\gt$  on the Fock module \cite{FHHSY} allows one to introduce a free complex parameter, the spectral parameter. Thus we denote the Fock module by $V(u)$, $u\in \bbC^*$. The basis elements of $V(u)$ are labelled by integer partitions. As a consequence we will encounter a number of manipulations involving partitions. Let us recap the basics of integer partitions and then turn to the discussion of the Fock module.

\subsection{Partitions}
Let $\MP$ be the set of all partitions $\m=(\m_1,\m_2,\dots)$,  $\m_j\geq \m_{j+1}$. As usual a partition $\m$ is identified with the associated Young diagram $\m$. The {\it length of a partition} $\m$, denoted $\ell(\m)$, is equal to the number of non-zero parts in $\m$. 
The {\it weight of a partition} $|\m|$ is given by the sum of all parts $|\m|=\m_1+\dots +\m_{\ell(\m)}$. The notation $\m\vdash j$ means that $\m$ is a partition of $j$, $|\m|=j$. By $\MP_j$ we denote the set of all partitions with weight $j$. The set $\MP_0$ contains one element which is the empty partition, denoted $\pzero$. The total number of elements in $\MP_j$ is denoted by $n(j)$.

Let $\m$ and $\n$ be two partitions such that the parts of $\n$ are given by a subset of parts of $\m$, then $\n$ is a {\it subpartition} of $\m$ which is denoted $\n\pin \m$. If $\n\pin\m$ then  the {\it complement} of $\m$ with respect to $\n$, denoted $\m\psub\n$, stands for the partition whose parts are given by removing parts of $\n$ from the set of parts of $\m$. The union of two partitions $\m\padd \n$ is the partition whose parts are given by the union of the parts of $\m$ and $\n$. The intersection of two partitions $\m$ and $\n$, denoted $\m\pint \n$, is the maximal subpartition $\s$ contained in both $\m$ and $\n$. If for two partitions $\m$ and $\n$ we have $\m_i\geq \n_i$ for all $i$ then $\m/\n$ denotes the sequence of non-negative integers $(\m_1-\n_1,\m_2-\n_2,\dots)$ which is identified with the skew Young diagram $\m/\n$. 

Let $r$ be a positive integer and $\mu$ a partition. Define the {\it part multiplicity function} $m_r(\m)$ which counts the number of parts in $\m$ equal to $r$. The {\it part multiplicity vector} $m(\m)$ is defined as $m(\m)=(m_1(\m),m_2(\m),\dots,m_{\m_1}(\m))$. Let us introduce short hand notations involving factorials of part multiplicities of a partition $\m$, we set 
\ba
&m(\m)!:=\prod_{a\in \m} m_a(\m)!.
\ea
We define an analogue of the binomial coefficient for situations when arguments are part multiplicity vectors of $\m$ and $\n$
\begin{align}
\label{binom0}
{\m \brack \n} :=
  \begin{cases}
\frac{m(\m)!}{m(\m\psub\n)!m(\n)!},       
&\quad \text{if } \n\pin\m, \\
0,  & \quad \text{if } \n\pnin\m.
  \end{cases}
\end{align}
All properties of binomial coefficients apply to (\ref{binom0}), in particular the symmetry
\begin{align}
\label{binom}
{\m \brack \n} =
{\m \brack \m\psub \n}.
\end{align}

\subsection{Fock representation}
\label{ssec:Fock}
In \cite{FHHSY} the authors construct the Fock representation for $\gt$ where $e(z)$ and $f(z)$ act by vertex operators. We will describe this representation here in the detail. The notation for the homomorphism taking elements of $\gt$ to their representations on $V(u)$ will be omitted for simplicity.   The central elements in $V(u)$ are specialized to $c=1$ and $\cp=0$. 

\subsubsection{Representation of the Heisenberg subalgebra}
The operators $h_{\pm r}$, $r\geq 0$, form the Heisenberg subalgebra $\MH_h$ due to (\ref{eq:hh}) 
\begin{align}
\label{eq:hhperp}
\MH_{h}=\{h_r,h_{-r}\}_{r>0}
,\qquad
[h_r,h_s]= \d_{r,-s} \frac{q^{r}-q^{-r}}{r \k_r}.
\end{align}
The homogeneous degree operator $q^d$ acts by
\begin{align}
\label{eq:dhperp}
q^{d}h_r q^{-d}= q^{-r} h_r. 
\end{align}
The operators $h_r$ commute with each other if their indices have the same sign. Therefore products of such operators can be written in any order. For a partition $\m$ we introduce the operators $h_{\m}$ and $h^*_{\m}$
\begin{align}
\label{def:hmu}
h_{\m}:=h_{\m_1}h_{\m_2}\dots = \prod_{a\in \m} h_{a},
\qquad 
h^*_{\m}:=h_{-\m_1}h_{-\m_2}\dots =  \prod_{a\in \m} h_{-a}.
\end{align}
Let $\ket{\mu}$ and $\bra{\mu}$, $\mu\in \MP$, denote the vectors of the Fock vector space and its dual, respectively.
Let $\ket{\emptyset}$ be the Fock space vacuum and $\bra{\emptyset}$ be the dual Fock space vacuum. The highest weight property is expressed as 
\begin{align}
\label{eq:hvac}
h_r\ket{\emptyset}=0,
\qquad 
\bra{\emptyset}h_{-r}=0,\qquad r>0. 
\end{align}
The action of the operators $h_r$ and $h_{-r}$, $r>0$, on the Fock space is given by the formulas:
\begin{align}
\label{eq:hact}
&h_r\ket{\mu}=m_r(\m) \frac{q^r-q^{-r}}{r \k_r}\ket{\mu\psub(r)},\\
\label{eq:hdact}
&h_{-r}\ket{\mu}= \ket{\mu\padd(r)},
\end{align}
while the action on the dual space is 
\begin{align}
\label{eq:hactdual}
&\bra{\mu}h_r= \bra{\mu\padd(r)},\\
\label{eq:hdactdual}
&\bra{\mu}h_{-r}=m_r(\m) \frac{q^r-q^{-r}}{r \k_r} \bra{\mu\psub(r)}.
\end{align}
With the help of (\ref{eq:hdact}) and (\ref{eq:hactdual}) we can view the operators $h^*_\m$ for $\m\in \MP$ as creation operators which generate all vectors of the Fock space from the vacuum and similarly $h_\m$ generate all vectors of the dual Fock space
\ba
h^*_{\m}\ket{\emptyset}=\ket{\m},
\qquad
\bra{\emptyset} h_{\m}=\bra{\m}.
\ea
From (\ref{eq:hact}) and (\ref{eq:hdactdual}) we see that if the part $r$ is absent from $\m$ then $m_r(\m)=0$ and the actions of $h_r$ and $h_{-r}$ on the corresponding ket and bra vectors give zero as required by (\ref{eq:hvac}).  The scalar product of two vacuum vectors is chosen to be normalized
\ba
\bra{\emptyset}\emptyset\rangle=1.
\ea
With this we compute the scalar product of two arbitrary vectors
\begin{align}
\label{eq:hsp}
\bra{\m}\n\rangle=\d_{\m,\n} \MN_{\m},
\qquad
\MN_{\m}:=
\prod_{a\in \m}\left(\frac{q^a-q^{-a}}{a\k_a}\right).
\end{align}
Where the delta symbol appears due to the commutation of $h_r$ and $h_{-s}$ and (\ref{eq:hvac}) and the normalization $\MN_{\m}$ is computed with (\ref{eq:hhperp})
\ba
\bra{\m}\m\rangle
&=
\bra{\emptyset}
h_\m
h^*_{\m}
\ket{\emptyset}
=
\bra{\emptyset}
\prod_{a\in\m}
h_a h_{-a}
\ket{\emptyset}\\
&=
\bra{\emptyset}
\prod_{a\in\m}
[h_a,h_{-a}]
\ket{\emptyset}
=
\prod_{a\in \m}\left(\frac{q^a-q^{-a}}{a\k_a}\right)
\bra{\emptyset} \emptyset \rangle=\MN_\m.
\ea
The identity matrix acting on the Fock space is given by
\begin{align}
\label{eq:idFock}
\text{Id}=
\sum_\a \frac{1}{\MN_\a} \ket{\a}\bra{\a}.
\end{align}
The formulas (\ref{eq:hvac})-(\ref{eq:idFock}) are required in order to write the $R$-matrix in the form (\ref{Rmatrix}).

\subsubsection{Representation of the currents}
Now we turn to the currents $e(z)$ and $f(z)$. Using the conventions of \cite{FJMM_BA} these operators are represented on $V(u)$ by vertex operators with the Heisenberg generators $h_{\pm r}$
\begin{align}
\label{VOe}
\begin{split}
e(z &)=\frac{1-q_2}{\k_1} u ~
\exp\left(\sum_{r=1}^{\infty}\frac{\k_r}{1-q_2^r} h_{- r}z^{ r}\right)
\exp\left(\sum_{r=1}^{\infty}\frac{q^r \k_r}{1-q_2^r} h_{ r}z^{- r}\right),\\
f(z)&=-\frac{1-q_2}{q_2 \k_1}u^{-1} 
\exp\left(-\sum_{r=1}^{\infty}\frac{q^r \k_r}{1-q_2^r} h_{- r}z^{ r}\right)
\exp\left(-\sum_{r=1}^{\infty}\frac{q^{2r} \k_r}{1-q_2^r} h_{ r}z^{- r}\right).
\end{split}
\end{align}
The powers of $u$ in these expressions give the principal grading of $e(z)$ and $f(z)$.  
From these vertex operators we can write the representation on the modes. 
The exponentials appearing in these formulas have the following expansion in terms of partitions:
\begin{align}
\label{eq:expsum}
&
\exp\left(\sum_{r=1}^{\infty} g_r z^{r}\right)
=
\sum_{j=0}^{\infty}
z^j
\sum_{\m \vdash j} \frac{1}{m(\m)!} \prod_{a\in \m}  g_a.
\end{align}
Using (\ref{eq:expsum}) we write the modes $e_n$ and $f_n$ of (\ref{VOe}) 
\begin{align}
\label{eq:epn}
\begin{split}
e_n &=
\frac{1-q_2}{\k_1} u ~
 \sum_{i=0}^{\infty}
 \sum_{\m \vdash i}
 \sum_{\n \vdash i+n}
 q^{i+n}  \frac{1}{m(\m)!m(\n)!} \prod_{a\in \m\padd\n}  \frac{\k_a}{1-q_2^a}\cdot
      h^*_{\m}	      h_{\n}
,\\
f_n &=
-\frac{1-q_2}{q_2 \k_1}u^{-1} 
 \sum_{i=0}^{\infty}
 \sum_{\m \vdash i}
 \sum_{\n \vdash i+n}
 (-1)^{\ell(\m)+\ell(\n)}q^{3i+2n}   
 \frac{1}{m(\m)!m(\n)!} \prod_{a\in \m\padd\n}  \frac{\k_a}{1-q_2^a}\cdot
      h^*_{\m}	      h_{\n}.
 \end{split}
\end{align}
The modes $\psi^{\pm}_{\pm j}$ are given by 
\begin{align}
\label{eq:psimn}
\begin{split}
\psi^{+}_j =
 \sum_{\m \vdash j}
\frac{1}{m(\m)!} \prod_{r\in \m}\k_r \cdot  h_{\m}
,\\
\psi^{-}_{-j}=
 \sum_{\m \vdash j}
\frac{1}{m(\m)!} \prod_{r\in \m}\k_r \cdot h^*_{\m}.
\end{split} 
\end{align}

\subsubsection{A PBW basis in the Fock representation}
We find it convenient to switch to a different normalization of the Heisenberg operators and introduce a related Heisenberg algebra $\MH_a$
\begin{align}
\label{acom}
\MH_{a}=\{a_r,a_{-r}\}_{r>0},\qquad 
[a_r, a_{-s}]
=
\d_{r,s}
q^r \frac{(1-q_1^r)(1-q_3^r)}{r}.
\end{align}
These operators are related to the operators of $\MH_h$  by
\begin{align}
\label{htoa}
&a_r =-q^r \frac{\k_r}{1-q_2^r}h_r, 
\qquad
&a_{-r} =q^r \frac{\k_r}{1-q_2^r}h_{-r}.
\end{align}
We define the elements  $a_\m, a^*_\m$ via the following formula
\begin{align}
\label{amu}
&a_\m = \frac{1}{m(\m)!} a_{\m_1}\dots a_{\m_{\ell(\m)}}, 
\qquad
&a^*_\m =\frac{1}{m(\m)!}  a_{-\m_1}\dots a_{-\m_{\ell(\m)}}.
\end{align}
The elements $a_{\m}$ and $a^*_{\m}$ are expressed through $h_{\m}, h^*_{\m}$ by 
\begin{align}
\label{agen}
a_{\m} =  \frac{(-1)^{\ell(\m)} q^{|\m|}}{m(\m)!} \prod_{a\in \m}  \frac{\k_a}{1-q_2^a}\cdot h_{\m}
, 
\qquad
a^*_{\m} =\frac{q^{|\m|}}{m(\m)!} \prod_{a\in \m}  \frac{\k_a}{1-q_2^a} \cdot h^*_{\m}.
\end{align}
These operators satisfy the following quadratic relations
\begin{align}
\label{eq:aa}
&a^*_\m a^*_\n ={\m\padd \n \brack \m} a^*_{\m\padd \n},
\qquad
a_\m a_\n ={\m\padd \n \brack \m} a_{\m\padd \n},
\end{align}
and the ordering relation
%
\begin{align}
\label{aquad2}
a_\n a^*_{\m} 
=
&\sum_{\l\subseteq \n\cap \m}
\xi_\l
a^*_{\m\psub\l}a_{\n\psub\l},
\qquad
\xi_\l:=
\frac{q^{|\l|}}{m(\l)!} \prod_{r\in \l}  \frac{(1-q_1^r)(1-q_3^r)}{r}.
\end{align}
Thanks to this relation we have the basis of the universal enveloping Heisenberg algebra spanned by $a^*_{\m} a_{\n}$. We can write (\ref{eq:epn})-(\ref{eq:psimn}) in terms of these operators. 
Writing coefficients explicitly we have 
\begin{align}
\label{eq:epa}
\begin{split}
e_n &=
\frac{1-q_2}{\k_1} u ~
 \sum_{\m,\n \in \MP}
 \d_{|\n|,|\m|+n} 
 (-1)^{\ell(\n)}
  q^{-|\m|}
a^*_{\m}
a_{\n}
,\\
f_n &=
-\frac{1-q_2}{q_2 \k_1}u^{-1} 
 \sum_{\m,\n \in \MP}
 \d_{|\n|,|\m|+n}
  (-1)^{\ell(\m)}
 q^{|\n|}  
a^*_{\m}
a_\n ,\\
\psi^{+}_j &=
 \sum_{\m \vdash j}
\prod_{r\in \m}(q^r-q^{-r}) \cdot 
a_{\m}
,\\
\psi^{-}_{-j} &=
 \sum_{\m \vdash j}
 (-1)^{\ell(\m)}
\prod_{r\in \m}(q^r-q^{-r}) \cdot 
a^*_{\m}.
\end{split}
\end{align}


\section{Coproduct using vertex operators}
\label{sec:coprod}
The identity (\ref{deltaUR}) must hold for all elements $g\in \gt$. We will make use of the factorization (\ref{KR}) and then write the full set of equations (\ref{deltaUR}) for the unknown element $\MRb$. After that we will specialize to the Fock module and obtain a linear system for the elements of $\MRb$ in the Fock basis in $V(u_1)\otimes V(u_2)$.
The main difficulty in this calculation is in solving the coproduct equation with $\D(e(z))$ and $\D(f(z))$.
We first solve the coproduct equations with $\D(a_r)$ and $\D(a_{-r})$ which shows some useful symmetries of the $R$-matrix. With these symmetries the main problem of solving the relations with $\D(e(z))$ and $\D(f(z))$ simplifies and reduces to a single equation.


\subsection{The coproduct relations for $\MRb$}
Let us conjugate the opposite coproduct by $\MK$ (\ref{MK}) and define
\begin{align}
\label{eq:Dt}
&\Dt(g):=\MK^{-1}\D'(g) \MK.
\end{align}
The relation (\ref{deltaUR}) turns in to an equation for $\MRb$
\begin{align}
\label{eq:RDt}
&\MRb\D(g)=\Dt(g) \MRb, \qquad g= \{e_k,f_k,h_r,h^*_r\}.
\end{align}
Let us compute $\Dt$.
\begin{lem}
For $g= \{e_k,f_k,h_r,h^*_r\}$ with $k\in \bbZ$ and $r>0$ we have 
\begin{align}
\label{eq:RDe}
\begin{split}
\Dt(e_k)
&=1\otimes e_k+  \sum_{j=0}^{\infty}  e_{k+j}\otimes q^{-c(k+j)}  \psi^-_{-j},\\
\Dt(f_k)
&=f_k \otimes 1 +  \sum_{j=0}^{\infty} q^{-c(k-j)} 
\psi^+_{j} \otimes f_{k-j},\\
\Dt(h_r)
&=h_{r} \otimes 1+ q^{-c r }  \otimes h_r, \\
\Dt(h_{-r})
&= h_{-r} \otimes q^{c r} + 1 \otimes h_{-r} .
\end{split}
\end{align}
\end{lem}
The derivations of these equations are presented in Appendix \ref{appdelta}. Notice that $\Dt(h_{\pm r})=\D(h_{\pm r})$ for $r>0$. Inserting (\ref{eq:De}) and (\ref{eq:RDe}) into (\ref{eq:RDt})
for $g= \{e_k,f_k,h_r,h^*_r\}$, respectively, gives us four equations for the element $\MRb$
\begin{align}
\label{eq:RDte}
\begin{split}
&[\MRb,1\otimes e_k] =  
\sum_{j=0}^{\infty} e_{k+j}\otimes q^{-c(k+j)}  \psi^-_{-j} \MRb
-
\MRb  \sum_{j=0}^{\infty} e_{k-j}\otimes q^{k c } \psi_j^{+} ,\\
&[\MRb,f_k\otimes 1] =
 \sum_{j=0}^{\infty} q^{-c(k-j)} 
\psi^+_{j} \otimes f_{k-j} \MRb
-
\MRb
\sum_{j=0}^{\infty} q^{ k c} \psi_{-j}^{-} \otimes  f_{k+j},\\
 &[\MRb,h_{r} \otimes 1+ q^{-c r } \otimes h_r] =0,\\
 &[\MRb, h_{-r} \otimes q^{c r} + 1 \otimes h_{-r}] =0.
\end{split}
\end{align}
Until now the results of this section hold for the algebra $\gt$. In the rest of the paper we will specialize to the Fock representation. In what follows we will use the Heisenberg operators $\{a_r,a_{-r}\}_{r>0}$ instead of $\{h_r,h_{-r}\}_{r>0}$, they are related through  (\ref{htoa}).


\subsection{The coproduct relations for $\MRb$ in the Fock representation}
In Section \ref{ssec:Fock} we described the Fock representation $V(u)$ with the spectral parameter $u$. It will become clear below that the $R$-matrix acting on $V(u_1)\otimes V(u_2)$ depends on the ratio of $u_1$ and $u_2$, so we can set $u_1=1$  and $u_2=u$ in what follows. We will denote the Fock $R$-matrix by $R(u)$ and the image of the matrix $\MRb$ by $\Rb(u)$.  The full $R$-matrix reads 
\begin{align}
\label{RFock}
&R(u)=\exp\left( \sum_{r\geq 1} \frac{r (1-q_2^{-r})}{(1-q_1^r)(1-q_3^r)}  a_{r} \otimes  a_{-r} \right)q^{-d\otimes 1-1\otimes d}
\Rb(u),
\end{align}
where the part $\MK$ in (\ref{MK}) is written in terms of the operators $a^*_r,a_r$ defined in (\ref{htoa}) and $\Rb(u)$ is the unknown part which we need to compute. 
As we explained, the basis in $V(u)$ is given by all operators $a^*_{\m}a_\n$, $\m,\n \in \MP$. Hence $\Rb(u)$ in the Fock representation $V(1)\otimes V(u)$ can be expanded in the basis given by the tensor product 
\begin{align}
\label{eq:Rbexp}
&\Rb(u)
     =
\sum_{\m,\r,\n,\s\in \MP}
(\Rb(u))_{\m,\r}^{\n,\s}
\,a^*_{\m}a_{\n}\otimes a^*_{\r}a_{\s},\\
&
(\Rb(u))_{\m,\r}^{\n,\s}=0,
\qquad |\m|+|\r|\neq |\n|+|\s|. 
\label{eq:Rbexpcons}
\end{align}
The restriction $ |\m|+|\r|=|\n|+|\s|$ comes from the requirement that all terms in $\MRb$ are neutral in the homogenous degree. From the identity term in $\MRb$ in (\ref{MR2}) we deduce the normalization
\begin{align}
\label{eq:00}
(\Rb(u))_{0,0}^{0,0}=1.
\end{align}

\subsection{The coproduct relation with elements of the Heisenberg subalgebra}
In this section we focus on the coproduct relation (\ref{eq:RDt}) with the elements $a_{\pm r}$ of the Heisenberg algebra $\MH_a$. As a consequence we will identify a set of independent functions in terms of which $\Rb(u)$ is expressed. First of all we rewrite the last two equations of (\ref{eq:RDte}) in terms of $a_{r}, a_{-r}$, $r>0$, using (\ref{htoa})
\begin{align}
\label{eq:DRa}
[\Rb(u),a_{r} \otimes 1+ q^{-r} 1 \otimes a_r]=0,\\
\label{eq:DRadag}
[\Rb(u),q^{r}  a_{-r} \otimes 1+ 1 \otimes a_{-r}]=0.
\end{align}
Let us introduce the coefficients $R_{\m,\n}(u)$ by the formula:
\begin{align}
\label{Pdef}
R_{\m,\n}(u):=(\Rb(u))_{\m,\pzero}^{\n,\emptyset}.
\end{align}
\begin{prop}
\label{lem:P}
Apart from a simple prefactor, the coefficients $(\Rb(u))_{\m,\r}^{\n,\s}$ of the R-matrix $\Rb(u)$ as defined in \eqref{eq:Rbexp} depend only on two partitions. Explicitly, equations (\ref{eq:DRa}) are solved by
\begin{align}
\label{PR}
(\Rb(u))_{\m,\r}^{\n,\s}
=(-1)^{\ell(\r)+\ell(\s)} q^{|\r|+|\s|} 
R_{\m\padd \r,\n\padd \s}(u).
\end{align}
\end{prop}
 The proof is given in Appendix~\ref{ap:proofProp:P}. We insert (\ref{PR}) into (\ref{eq:Rbexp}) and rewrite the summation
\begin{align}
\label{eq:RPexp}
&\Rb(u)
     =
\sum_{\m,\n\in \MP}
R_{\m,\n}(u)
\sum_{\mb\pin \m}
(-1)^{\ell(\mb)} q^{|\mb|} 
a^*_{\m\psub\mb}\otimes a^*_{\mb}
\sum_{\nb\pin \n}
(-1)^{\ell(\nb)} q^{|\nb|} 
a_{\n\psub\nb}\otimes a_{\nb}.
\end{align}
This suggests that we can write $\Rb$ using only one set of Heisenberg operators as apposed to using two sets for the two tensor product factors $\MH_a\otimes 1$ and $1 \otimes \MH_a$. Define new operators $b_{\pm r}$ and $c_{\pm r}$, with $r>0$, by
\begin{align}
b_{r}&:=a_r\otimes 1 + q^{-r} 1\otimes a_r,\qquad
b_{-r}:=a_{-r}\otimes 1 + q^{-r} 1\otimes a_{-r},\\
\label{eq:ca}
c_{r}&:=q^{-r} a_r\otimes 1 -  1\otimes a_r,\qquad
c_{-r}:=q^{-r} a_{-r}\otimes 1 -  1\otimes a_{-r}.
\end{align}
The operators $b_r$ commute with all operators $c_s$ for all $r,s$. Their commutation relations are given by
\begin{align}
\label{bcom}
[b_r, b_{-s}]
&=
\d_{r,s}
 \frac{(q^r+q^{-r})(1-q_1^r)(1-q_3^r)}{r},\qquad
 [c_r, c_{-s}]
=
\d_{r,s}
 \frac{(q^r+q^{-r})(1-q_1^r)(1-q_3^r)}{r}.
\end{align}
The building blocks of the bases are given by
\begin{align}
\label{bmu}
b_\m = \frac{1}{m(\m)!} b_{\m_1}\dots b_{\m_{\ell(\m)}}, 
&\qquad
b^*_\m =\frac{1}{m(\m)!}  b_{-\m_1}\dots b_{-\m_{\ell(\m)}},\\
\label{cmu}
c_\m = \frac{1}{m(\m)!} c_{\m_1}\dots c_{\m_{\ell(\m)}} , 
&\qquad
c^*_\m =\frac{1}{m(\m)!}  c_{-\m_1}\dots c_{-\m_{\ell(\m)}}.
\end{align}
The inverse transform is given by 
\begin{align}
\label{eq:a1tobc}
a_r\otimes 1&=\frac{1}{q^r+q^{-r}}\left(q^r b_r+ c_r \right),\qquad
a_{-r}\otimes 1=\frac{1}{q^r+q^{-r}}\left(q^r b_{-r}+ c_{-r} \right),\\
\label{eq:a2tobc}
1\otimes a_r&=\frac{1}{q^r+q^{-r}}\left( b_r-q^r c_r \right),\qquad
1\otimes a_{-r}=\frac{1}{q^r+q^{-r}}\left( b_{-r}-q^r c_{-r} \right).
\end{align}

\begin{prop}
The R-matrix $\Rb$ can be written in terms of the operators $c_\mu$ in the following way
\begin{align}
\label{eq:Rbarbc}
&\Rb(u)
     =
\sum_{\m,\n\in \MP}
q^{|\m|+|\n|} R_{\m,\n}(u)
c^*_\m c_\n.
\end{align}
\end{prop}
\begin{proof}
Using (\ref{eq:ca}) we can write $c^*_\m$ and $c_\m$ in terms of $a^*_\m$ and $a_\m$
\begin{align*}
c_{\m}=\sum_{\s\pin \m} (-1)^{\ell(\s)}q^{-|\m|+|\s|} a_{\m\psub\s} \otimes a_{\s} ,\qquad
c^*_{\m}=\sum_{\s\pin \m}(-1)^{\ell(\s)} q^{-|\m|+|\s|} a^*_{\m\psub\s} \otimes a^*_\s.
\end{align*}
These formulas match with the sums over $\mb$ and $\nb$ in (\ref{eq:RPexp}), after this (\ref{eq:Rbarbc}) follows. 
\end{proof}

\subsection{Reduction of the coproduct equation}
Let us introduce two vertex operators 
\begin{align}
\label{def:phipm}
\phi^-(z):= 
\exp\left(\sum_{r=1}^{\infty}\frac{1}{q^r+q^{-r}}c_{-r}z^{ r}\right),
\qquad
\phi^+(z):= 
\exp\left(\sum_{r=1}^{\infty}\frac{1}{q^r+q^{-r}}c_{r}z^{-r}\right).
\end{align}
Using the operators  $\phi^{\pm}(z)$ we can rewrite the coproduct equation with $\D(e(z))$ and $\D(f(z))$. This computation is presented in detail in Appendix \ref{sec:VOeq}.   The result is given in the following proposition.
\begin{prop}
\label{propRFE}
The operator $\Rb(u)$ satisfies
\begin{align}
\label{eq:RFE}
u 
[\Rb(u),\phi^-(z)^{-1}\phi^+(z q^{-1})]
&=
  \phi^-(z q^{2})  \phi^+(z q)^{-1}\Rb(u)- \Rb(u)   \phi^-(z q^{-2}) \phi^+(z q^{-3})^{-1}.
\end{align}
\end{prop}
It follows from the calculations in Appendix \ref{sec:VOeq} that both coproduct equations with $\D(e(z))$ and $\D(f(z))$ lead to (\ref{eq:RFE}). Therefore  this single equation takes into account all coproduct equations. Equation (\ref{eq:RFE}) appeared in the paper \cite{Fukud} in which the authors use this equation to compute the first term in the expansion of $\Rb(u)$ in the spectral parameter. In the following our aim is to rewrite (\ref{eq:RFE}) by restoring the Fock action of the operators of $\gt$. The resulting equation is a vector equation in the tensor product of two Fock spaces.

Consider the Fock spaces $V$ and $V^*$ spanned by the vectors $\ket{c_\m}$ and $\ket{c_\m}^*$ respectively,
\begin{align}
\label{eq:VHiso}
c^*_\m\ket{\pzero}=\ket{c_\m}
,\qquad c_\m \ket{\pzero}^*=\ket{c_\m}^*.
\end{align}
The universal enveloping algebra of $\MH_c$ as a vector space is spanned by vectors $\{c^*_\m c_\n\}_{\m,\n\in \MP}$. This vector space can be identified with $V\otimes V^*$ by the assignment 
\begin{align}
\label{eq:iso0}
\iota':~\MH_c \rightarrow V\otimes V^*,\qquad \iota'(c^*_\m c_\n) =  \ket{c_\m}\otimes \ket{c_\n}^*.
\end{align}
A second ingredient is an automorphism of $\MH_c$:
\begin{align}
\label{eq:autu}
\ve: \quad c_{\pm r} \mapsto \mp c_{\mp r}, \qquad r>0,
\end{align}
which when combined with the above provides an identification of the vector space of $\MH_c$ with the vector space $V\otimes V$
\begin{align}
\label{eq:iso1}
\begin{split}
&\iota=(1\otimes \ve)\iota':~\MH_c \rightarrow V\otimes V
,\\
&\iota(c^*_\m c_\n) =(-1)^{\ell(\n)}   \ket{c_\m}\otimes \ket{c_\n}.
\end{split} 
\end{align}

\begin{lem}
\label{lem:phi}
The assignment $\iota$ extends to an isomorphism $\iota:~\MH_c  \stackrel{\simeq}{\rightarrow}  V\otimes V$. Using the identification \eqref{eq:iso1} we have the following relations:
\begin{align}
\label{eq:iphi1}
\begin{split}
\iota(\phi^+(z)c^*_\m c_\n )
&=
(-1)^{\ell(\n)} \phi^+(z) \ket{c_\m} \otimes  \phi^-(z^{-1})^{-1} \ket{c_\n},\\
\iota(\phi^-(z)c^*_\m c_\n )
&=(-1)^{\ell(\n)}  \phi^-(z) \ket{c_\m}  \otimes \ket{c_\n},\\
\iota(c^*_\m c_\n  \phi^+(z))
&=(-1)^{\ell(\n)}  \ket{c_\m}  \otimes \phi^-(z^{-1})^{-1} \ket{c_\n},\\
\iota(c^*_\m c_\n  \phi^-(z))
&= (-1)^{\ell(\n)}  \phi^-(z) \ket{c_\m} \otimes \phi^+(z^{-1})^{-1} \ket{c_\n}.
\end{split}
\end{align}
\end{lem}
\begin{proof}
Let $\MH_c^-=\{c^*_\m \}_{\m\in \MP}$ and $\MH_c^+=\{c_\n \}_{\n\in \MP}$ be the negative and positive parts of $\MH_c$.  We have the following invariant actions by vertex operators
\begin{align}
\label{eq:phiH1}
\phi^{+}(z)\MH_c^+ \pin \MH_c^+,
&\qquad
\phi^{-}(z)\MH_c^+ \phi^{-}(z)^{-1} \pin \MH_c^+,\\
\label{eq:phiH2}
\phi^{-}(z)\MH_c^- \pin \MH_c^-,
&\qquad
\phi^{+}(z)^{-1}\MH_c^- \phi^{+}(z) \pin \MH_c^-.
\end{align}
Consider now the left hand side of the first equation in \eqref{eq:iphi1},
 \ba
  \iota(\phi^+(z)c^*_\m c_\n ) = \iota(\phi^+(z)c^*_\m \phi^+(z)^{-1} \phi^+(z) c_\n ) .
  \ea 
 Because $\phi^+(z)c^*_\m \phi^+(z)^{-1} \in \MH_c^-$ and $\phi^+(z) c_\n \in\MH_c^+$, and because of the linearity of $\iota'$, we can apply \eqref{eq:iso0},
\ba
 \iota'(\phi^+(z)c^*_\m \phi^+(z)^{-1} \phi^+(z) c_\n )  = \phi^+(z) c^*_\m \phi^+(z)^{-1}  \ket{\pzero} \otimes \phi^+(z) \ket{c_\n}^*.
 \ea
 Note now that $\phi^+(z)$ acts trivially on $\ket{\pzero}$ and so we have
\ba
\phi^{+}(z)c^*_\m \phi^{+}(z)^{-1} \ket{\pzero} = \phi^{+}(z) c_\m^* \ket{\pzero} = \phi^{+}(z) \ket{c_\m} .
\ea
Putting everything together and applying $\ve$ we thus find
\ba
\iota(\phi^+(z)c^*_\m c_\n ) =(-1)^{\ell(\n)} \phi^+(z) \ket{c_\m} \otimes  \phi^-(z^{-1})^{-1} \ket{c_\n}.
\ea
The other relations in the lemma follow in a similar way.
\end{proof}

By considering $\Rb(u)$ as a vector in the space $V\otimes V$, the operator equation (\ref{eq:RFE}) represents an equation where vertex operators  $\phi^{\pm}(z)$ act on the vector $\Rb(u)$. 
\begin{defn}
Introduce the vector $\DR(u)$ in $V\otimes V$
\begin{align}
\label{eq:Rvec}
\DR(u):=\sum_{\m,\n} (-1)^{\ell(\n)} q^{|\m|+|\n|} R_{\m,\n}(u) \ket{c_\m}\otimes \ket{c_\n}.
\end{align}
\end{defn}
Clearly we have
\ba
\iota(\Rb(u))= \DR(u).
\ea
Using Lemma~\ref{lem:phi} we rewrite (\ref{eq:RFE}) as an equation in $V\otimes V$. After this we recover the action of the vertex operators of $\gt$. 
\begin{defn}
Define two operators
\begin{align}
\label{eq:Phip}
\Phi^+(z)
&:= \phi^+(z q^{-1}) \otimes   \phi^-(z^{-1} q^{3})^{-1} \phi^-(z^{-1} q)^{-1},\\
\label{eq:Phim}
\Phi^-(z)
&:= \phi^-(z)\phi^-(z q^{-2})\otimes   \phi^+(z^{-1} q^{2})^{-1}.
\end{align}
Let $\tau$ be the operation of transposition $\tau (a\otimes b)= b\otimes a$, then we define 
\begin{align}
\label{eq:Psipm}
\Psi^+(z):=\tau \Phi^+(z^{-1}),
\qquad 
\Psi^-(z):=\tau \Phi^-(z^{-1}).
\end{align}
\end{defn}
Recall the vertex operator representation (\ref{VOe}) on the space $V(u)$. Let $\{\ed(z),\fd(z),\psid^{\pm}(z)\}$ be another copy of the vertex operators acting on the space $V(1)$. 
\begin{prop}
\label{prop:copef}
Let $s:=\k_1/(1-q_2)=(1-q_1)(1-q_3)$. The vector $\DR(u)$ satisfies 
\begin{align}
\label{eq:DRVO}
u  \left[
s q^2~ 1\otimes  \fd(z^{-1} q) +  \Phi^+(z)\right]\DR(u) 
&= \left[
\Phi^-(z) - s~ \ed(z q)\otimes \psid^{-}(z^{-1})
 \right]\DR(u),\\
\label{eq:DRVOd}
u\left[ s q^2  \fd(z q) \otimes  1+\Psi^+(z) \right]\DR(u) 
&=
\left[\Psi^-(z) -
s~\psid^-(z )\otimes \ed(z^{-1}q) \right]\DR(u)
.
\end{align}
\end{prop}
The proof is given in Appendix \ref{ap:proofcop}. 
\begin{cor}
\label{cor:sym}
The vector $\DR(u)$ satisfies:
\begin{align}
\label{eq:Rsym}
\DR(u)=\tau (\DR(u)).
\end{align}
\end{cor}
\begin{proof}
This symmetry follows after we apply $\tau$ to (\ref{eq:DRVO}) and compare it to (\ref{eq:DRVOd}) taking into account (\ref{eq:Psipm}). 
\end{proof}
\begin{thm}
\label{eq:thmzero}
The vector $\DR(u)$ satisfies:
\begin{align}
\label{eq:DRVOzero}
u  \left(
s q^2~ 1\otimes  \fd_0 + 1\otimes1\right)\DR(u) 
= \left(
 1\otimes1 - s  \sum_{j\geq 0} q^{j} \ed_{-j}\otimes  \psid^-_{-j}
 \right)\DR(u). 
\end{align}
\end{thm}
\begin{proof}
This equation  is derived from (\ref{eq:DRVO}) by expanding in $z$ and selecting the coefficients of $z^0$. Indeed, notice that from (\ref{eq:Phip}) and (\ref{eq:Phim}) it follows that the $z$ expansion of $\Phi^+(z)$ has only terms $z^j$ with $j\leq 0$, and the $z$ expansion of $\Phi^-(z)$ has only terms $z^j$ with $j\geq 0$, therefore their constant terms in $z$ are identity operators. 
\end{proof}

\section{The $R$-matrix and Macdonald polynomial theory}
\label{sec:MacdonaldR}
The Heisenberg algebra as a graded vector space is isomorphic to the space of symmetric functions. This isomorphism takes the Heisenberg operators $c^*_\m$ and replaces them with a modified power sum symmetric function and the Heisenberg operators $c_\m$ with a derivative operator. Under this isomorphism $\DR(u)$ becomes a symmetric function $R(x,y;u)$ in two alphabets $x=(x_1,x_2,\dots)$ and $y=(y_1,y_2,\dots)$. The operators $\ed_0$ and $\fd_0$ are matched with Macdonald difference operators. The second term on the right hand side of (\ref{eq:DRVOzero}) gives rise to a ratio of Cauchy kernels. Thus we rephrase (\ref{eq:DRVOzero}) in terms of basic operators of the Macdonald theory. It is natural to expand $R(x,y;u)$ in the basis of Macdonald polynomials. For the coefficients of this expansion we derive a recursive formula. This recursive formula can be implemented on a computer giving the coefficients of $R(x,y;u)$ corresponding to partitions of small sizes. Using the transition coefficients from the Macdonald polynomials to the power sum polynomials in (\ref{Rmac})  we can explicitly compute matrix elements of the full $R$-matrix given in (\ref{Rmatrix}).

\subsection{Equation (\ref{eq:DRVOzero}) in terms of Macdonald operators}
We start by explaining the isomorphism to the space of symmetric functions. For details on the theory of symmetric function we refer to \cite{MacdBook}. The parameters $q,t$ which we used throughout this paper are not the same as those from the Macdonald theory. Instead we use $\qq,\tt$; for example the Macdonald polynomials are denoted $P_\l(x;\qq,\tt)$. Let $\qq,\tt$ be indeterminates and set $\mathbb{F}=\mathbb{Q}(\qq,\tt)$. The relationship between $\qq,\tt$ and our parameters $q$ and $t$ (also $q_1,q_3$) is given by 
\begin{align}
\label{eq:qt}
\qq = q t^{-1}=q_3^{-1},\qquad \tt = q^{-1} t^{-1}=q_1.
\end{align}

The ring of symmetric functions is denoted $\Lambda$, set $\Lambda_{\mathbb{F}}:=\Lambda\otimes_{\mathbb{Z}}\mathbb{F}$. Let $p_r(x)$ be the power sum symmetric function in the alphabet $x=(x_1,x_2\dots)$
\begin{align}
\label{eq:powersum}
p_{r}(x):=\sum_i x_i^r,
\qquad
p_\m(x)=p_{\m_1}(x)\dots p_{\m_{\ell(\m)}}(x).
\end{align}
The functions $p_\m(x)$ define a basis in $\Lambda_{\mathbb{F}}$. Macdonald's scalar product on $\Lambda_{\mathbb{F}}$ is given by 
\begin{align}
\label{eq:Mdscalar}
\braket{p_\l}{p_\m}=\d_{\l,\m} m(\l)!\prod_{a\in\l}a \frac{1-\qq^{a}}{1-\tt^a}. 
\end{align}
Let $\at_{\pm r}$ be Heisenberg operators generating the Heisenberg algebra $\MH_{\at}$ with the commutator 
\ba
[\at_r,\at_{-s}]=\d_{r,s} r \frac{1-\qq^r}{1-\tt^r}, \qquad {r,s>0}. 
\ea
The basis in the Fock space generated by $\at_{\pm r}$ is given by 
\ba
\ket{\at_\m}=\at^*_{\m}\ket{\emptyset},
\qquad
\bra{\at_\m}=\bra{\emptyset} \at_{\m},
\ea
where $a_\m =a_{\m_1}\dots a_{\m_{\ell(\m)}}$ and $a^*_\m =a_{-\m_1}\dots a_{-\m_{\ell(\m)}}$.
As graded vector spaces, the space of symmetric functions and the Fock space $V$ are isomorphic. Under this isomorphism, which we denote by $\pi$, we have
\begin{align}
\label{eq:autpi0}
\pi: \at^*_\l\mapsto p_\l,
\end{align}
and the two scalar products match 
\ba
\braket{p_\l}{p_\m}=\braket{\at_\l}{\at_\m}.
\ea
Let $\MH_{c}$ be the Heisenberg algebra with generators $c_{\pm r}$. A mapping between $\MH_{\at}$ and $\MH_{c}$ is given by a family of isomorphisms
\begin{align}
\label{eq:hatilde}
c_{-r}\rightarrow \g_r \frac{(1-q_1^r)(1-q_3^r)}{r} \at_{-r},
\qquad
c_r\rightarrow \g_r^{-1} \frac{(q^r+q^{-r})(1-q_1^r)}{r(1-q_3^{-r})}\at_r,
\end{align}
where $\g_r$ is a gauge factor which is related to the {\it plethystic substitution} in the theory of symmetric functions. 
The isomorphism between the Fock space generated by $\MH_c$ and symmetric functions is given by (\ref{eq:autpi0}) and (\ref{eq:hatilde}). This isomorphism relates $c_{-r}$ with the power sums $p_r$ and $c_{r}$ with derivatives with respect to $p_r$, namely
\begin{align}
\label{eq:piso}
\pi_\g: \qquad 
c_{-r}\mapsto \g_r \frac{(1-q_1^r)(1-q_3^r)}{r} p_r,
\quad
c_{r}\mapsto \g_r^{-1} \frac{-q_2^{r} (q^{r}+q^{-r}) (1-q_1^r) }{r (1-q_3^r)} \frac{\partial}{\partial p_r}. 
\end{align}
In order to recover the action of Macdonald operators in what follows we make the following choice of $\g_r$ and introduce the corresponding isomorphism 
\begin{align}
\label{eq:gam}
\g:\quad \g_r=\frac{q^{-r}}{1-q_3^{r}},
\qquad 
\pi:=\pi_\g.
\end{align}

A key object of study in Macdonald theory \cite{MacdBook} is the set of symmetric Macdonald polynomials $P_\l(x;\qq,\tt)$, $\l\in\MP$. The Macdonald polynomials are eigenvectors of the Macdonald operators and also form a basis in the ring of symmetric functions $\Lambda_{\mathbb{F}}$.  It is known \cite{Sh} that the isomorphism $\pi$ takes $\ed_0$ and $\fd_0$ in the Fock representation to operators which can be written in terms of the Macdonald operator which is denoted by $E$ in  \cite{MacdBook}. Let us introduce the coefficients $\e_\l(u,v)$:
\begin{align}
\label{eq:Mev}
\begin{split}
&\e_\l(u,v):=(v^{-1}-1)\sum_{i=1}^{\ell(\l)}(u^{\l_i}-1)v^{i},\\
&\e_\l: = \e_\l(\qq^{-1},\tt),\qquad
\eb_\l: = \e_\l(\qq,\tt^{-1}).
\end{split}
\end{align}
If we let $\ed_0$ act on $V(1)$ and $\fd_0$ act on $V(1)$ then, under the isomorphism $\pi$, we define 
\begin{align}
\label{eq:EEbar}
E : = s \pi( \ed_0)  -1 ,
\qquad
\Eb : = -s q^2\pi( \fd_0) -1.
\end{align}
The operators $E$ and $\Eb$ that act in the space of symmetric functions with the alphabet $x=(x_1,x_2,\dots)$ will be given the index $x$, i.e. we will write $E_x$ and $\Eb_x$. 
We identify two Macdonald operators\footnote{In \cite{MacdBook} (Ch. 6, \S 4) Macdonald introduces the operator denoted by $E$. This operator matches with  our $(\qq-1)^{-1}\bar{E}$ operator.}
which satisfy the following eigenvalue equations, 
\begin{align}
\label{eq:EPact}
\begin{split}
E_x P_\l(x;\qq,\tt) &= \e_{\l} P_\l(x;\qq,\tt),\\
\bar{E}_x P_\l(x;\qq,\tt) &=\eb_{\l} P_\l(x;\qq,\tt).
\end{split}
\end{align}

Since we work with the tensor product of two spaces we need two copies of symmetric function spaces: one with the alphabet $x=(x_1,x_2,\dots)$ and the second with $y=(y_1,y_2,\dots)$. 
The last ingredient is the Cauchy kernel for Macdonald polynomials. First, we introduce the deformed Pochhammer symbol 
\begin{align}
(a;\qq)_{\infty}:=\prod_{r=0}^{\infty}(1-a \qq^r).
\end{align}
The Cauchy kernel is a symmetric function in two alphabets $x$ and $y$ which is defined as follows 
\begin{align}
\label{eq:CK}
\Pi(x,y):=
\exp\left( \sum_{r\geq 1}\frac{(1-\tt^r)}{r(1-\qq^{r})} p_r(x)p_r(y)  \right)
=\frac{(\tt x y;\qq)_\infty}{(x y;\qq)_\infty}.
\end{align}
\begin{defn}
Using the isomorphism  $\pi$ we introduce the symmetric function $R(x,y;u)$:
\begin{align}
\label{eq:Rsymg}
R(x,y;u):=(\pi\otimes \pi)\DR(u).
\end{align}
\end{defn}
\begin{thm}
\label{thm:DPR}
The symmetric function $R(x,y;u)$ satisfies the equation
\begin{align}
\label{eq:DPiR0}
E_x\Pit(x,y) R(x,y;u)
=u~
\Pit(x,y) \bar{E}_y   R(x,y;u),
\end{align}
where the modified Cauchy kernel $\Pit(x,y)$ is defined by
\begin{align}
\label{eq:ModCK}
\Pit(x,y):=\Pi(q x,\qq y)^{-1} \Pi(q x,\tt y).
\end{align}
\end{thm}
\begin{proof}
Let us introduce the following operator
\ba
T:=\exp\left(-\sum_{r\geq 1}r   \frac{q^r(1-q_2^r)}{(1-q_1^r)(1-q_3^r)} c_{-r} \otimes c_{-r} \right).
\ea
Similarly to the derivation in Appendix \ref{eq:subcalc} we can show that
\ba
&T^{-1} (\ed_0\otimes 1) T
= \sum_{j=0}^{\infty} q^j \ed_{-j}\otimes \psid^-_{-j}.
\ea
Using the operator $T$ we factorise the second term on the right hand side of (\ref{eq:DRVOzero}) and then we multiply both sides by $T$
\begin{align}
\label{eq:DRTzero}
u  T\left(
s q^2~ 1\otimes  \fd_0 + 1\otimes1\right)\DR(u) 
= \left(
 1\otimes1 - s~ \ed_0\otimes 1  \right)T \DR(u). 
\end{align}
The next step is to rewrite this equation using symmetric functions and Macdonald operators. Under the isomorphism $\pi$ we have
\ba
(\pi\otimes \pi)T
&=
\exp\left(-\sum_{r\geq 1}\frac{(1-q_1^{r})(1-q_2^{r})}{r(1-q_3^{r})} q^{-r} p_r(x)p_r(y)  \right)=
\exp\left(-\sum_{r\geq 1}\frac{q^{r}(1-\tt^r)(\qq^{r}-\tt^{r})}{r (1-\qq^{r})} p_r(x)p_r(y)  \right).
\ea
This can be written using (\ref{eq:CK}) and (\ref{eq:ModCK})
\ba
(\pi\otimes \pi)T=\Pi(q x,\qq y)^{-1} \Pi(q x,\tt y)=\Pit(x,y).
\ea
Applying $(\pi\otimes \pi)$ to (\ref{eq:DRTzero}) and using (\ref{eq:EEbar}) we arrive at (\ref{eq:DPiR0}).
\end{proof}
A detailed analysis of equation (\ref{eq:DPiR0}) within the framework of the theory of symmetric functions will be a subject of a separate work. In the next subsection we expand $R(x,y;u)$ in the basis of Macdonald polynomials and derive a recursive formula for the expansion coefficients. This recursive formula is suitable for computing explicitly the matrix elements labelled by partitions with small weights.

\subsection{Recurrence relation}
We start with (\ref{eq:DRVOzero}) and using the isomorphism with symmetric functions we rewrite all the entries of the equation using symmetric functions operators. As noted before we consider $R(x,y;u)$ in the basis of Macdonald polynomials.
\begin{defn}
Let $\l$ and $\m$ be two partitions of the same weight, define a new set of coefficients $L_{\l,\m}(u)$ by the expansion 
\begin{align}
\label{eq:Lcoef}
R(x,y;u)=\sum_{\l,\m} L_{\l,\m}(u) P_\l(x;\qq,\tt)P_\m(y;\qq,\tt),
\end{align}
and for $|\l|\neq|\m|$ we set $L_{\l,\m}(u):=0$. 
\end{defn}
\begin{thm}
The coefficients $L_{\a,\b}(u)$ are symmetric 
\begin{align}
\label{eq:Lsym}
L_{\a,\b}(u)=L_{\b,\a}(u),
\end{align}
and satisfy a recursive formula
\begin{align}
\label{eq:Lrec}
L_{\a,\b}(u)
&=
\sum_{\l,\m} 
L_{\a/\l,\b/\m}(u)
 L_{\l,\m}(u),
 \end{align}
 with the initial condition $ L_{\pzero,\pzero}(u)=1$. The sum in (\ref{eq:Lrec}) runs over all $\l,\m$ such that both $\a/\l$ and $\b/\m$ are non-empty skew partitions.
\end{thm}
\begin{proof}
The symmetry (\ref{eq:Lsym}) follows from Corollary \ref{cor:sym} which says that $\DR(u)=\tau (\DR(u))$ where we recall that $\tau$ transposes the two vector spaces. On the right hand side of (\ref{eq:Lrec}) we have a sum over $\l,\m$ whose weights are strictly smaller than the weights of $\a$ and $\b$ therefore, given $L_{\a/\l,\b/\m}(u)$, equation (\ref{eq:Lrec}) provides us a recursive formula for computing $L_{\a,\b}(u)$ starting with the initial condition $ L_{\pzero,\pzero}(u)=1$.  The existence of this recurrence relation will be proved constructively in the rest of this section. The coefficients $L_{\a/\l,\b/\m}(u)$ are given in Proposition \ref{prop:Lbranch}. 
\end{proof}

We apply $\pi$ to (\ref{eq:DRVOzero}) and expand $R(x,y;u)$ in the Macdonald basis (\ref{eq:Lcoef}). The following lemma gives the action of $\pi(\psid^-_{-j})$ and $\pi(\ed_{-j})$ on Macdonald polynomials. 
\begin{lem}
The action of $\pi(\psid^-_{-j})$ and $\pi(\ed_{-j})$, for $j>0$, in the Macdonald basis is given by 
\begin{align}
\label{eq:psimpiact}
&\pi(\psid^{-}_{-j})P_\m(x;\qq,\tt) = \sum_\l a_{\l/\m}(\qq,\tt) P_\l(x;\qq,\tt),\\
\label{eq:epiact}
&\pi(\ed_{-j})P_\m(x;\qq,\tt) =\sum_\l \frac{a_{\l/\m}(\qq,\tt)}{(1-q^{-2})(\eb_\l-\eb_\m)} P_\l(x;\qq,\tt),
\end{align}
where $a_{\l/\m}(\qq,\tt)$ can be computed from the following formula:
\begin{align}
\label{eq:aPexpand}
\prod_i \frac{\left(q-\tt x_i z\right) \left(q-\qq \tt^{-1} x_i z\right)}{(q-x_i z) \left(q-\qq x_i z\right)}\cdot
P_\m(x;\qq,\tt)
=
\sum_{\l}z^{|\l|-|\m|}a_{\l/\m}(\qq,\tt)
P_\l(x;\qq,\tt).
\end{align}
\end{lem}
\begin{proof}
Due to the commutation relation in the fourth line of (\ref{eq:hh}) we can express the action of $\pi(\ed_{-j})$ through the action of $\pi(\psid^-_{-j})$. Indeed, having a suitable set of coefficients $a_{\l,\m}(\qq,\tt)$ and $b_{\l,\m}(\qq,\tt)$ we can write
\begin{align}
\label{eq:psiacoef}
&\pi(\psid^{-}_{-j})P_\m(x;\qq,\tt) = \sum_\l a_{\l,\m}(\qq,\tt) P_\l(x;\qq,\tt),\\
\label{eq:ebcoef}
&\pi(\ed_{-j})P_\m(x;\qq,\tt) =\sum_\l b_{\l,\m}(\qq,\tt) P_\l(x;\qq,\tt).
\end{align}
Using the commutation relation in the fourth line of (\ref{eq:hh}) we have 
\begin{align}
 \pi(\psid^{-}_{-j})P_{\m}(x;\qq,\tt)&= \k_1 [\pi(f_{0}),\pi(\ed_{-j})]P_{\m}(x;\qq,\tt)
= \k_1 \left(\pi(f_{0})-\frac{1+\eb_\m}{-s q^2}\right)\pi(\ed_{-j}) P_{\m}(x;\qq,\tt) \nonumber \\
& = \frac{\k_1}{-s q^2}\sum_\l b_{\l,\m}(\qq,\tt) \left(\eb_\l-\eb_\m\right) P_\l(x;\qq,\tt) ,\label{eq:pipsij0}
\end{align}
where we used (\ref{eq:EEbar}), (\ref{eq:EPact}) and (\ref{eq:ebcoef}). Recall that $s=\k_1/(1-q_2)$. The left hand side in the first line in (\ref{eq:pipsij0}) can be expanded in Macdonald polynomials (\ref{eq:psiacoef}), then by matching coefficients of polynomials we obtain
\begin{align}
\label{eq:abrel}
 b_{\l,\m}(\qq,\tt) =\frac{a_{\l,\m}(\qq,\tt)}{(1-q^{-2})(\eb_\l-\eb_\m)}.
\end{align}
We thus only need to compute $a_{\l,\m}(\qq,\tt)$. The current $\psid^{-}(z)$ can be written in terms of the algebra $\MH_c$ using (\ref{psiexp}) and the relation between $\MH_c$ and $\MH_h$
\ba
\psid^{-}(z) &= \exp\left(\sum_{r=1}^{\infty} (1-q_2^{r}) c_{- r} z^{ r}\right).
\ea
Applying to this expression the isomorphism (\ref{eq:piso})
\begin{align}
\label{pipsiexpX}
\pi(\psid^{-}(z))
=\exp\left(\sum_{r=1}^{\infty} \frac{1}{r}(1-\tt^r)(1-\qq^r \tt^{-r})q^{-r} p_r(x) z^{ r}\right) 
&=\prod_i \frac{\left(q-\tt x_i z\right) \left(q-\qq \tt^{-1} x_i z\right)}{(q-x_i z) \left(q-\qq x_i z\right)},
\end{align}
and using (\ref{eq:psiacoef}) it follows that $a_{\l,\m}=a_{\l/\m}$ in (\ref{eq:aPexpand}).
\end{proof}
Let us make a remark on the relation of $a_{\l/\m}$ with the skew Macdonald polynomials. We recall the skew Cauchy identity from \cite{MacdBook}
\begin{align}
\label{eq:skewC}
\Pi(x,y)\sum_\s
P_{\m/\s}(x;\qq,\tt)
Q_{\l/\s}(y;\qq,\tt)
=
\sum_\r
P_{\r/\l}(x;\qq,\tt)
Q_{\r/\m}(y;\qq,\tt).
\end{align}
In this equation set $\l=\pzero$, then the sum over $\s$ on the left hand side collapses
\begin{align}
\label{eq:skewC2}
\exp\left( \sum_{r\geq 1}\frac{(1-\tt^r)}{r(1-\qq^{r})} p_r(x)p_r(y)  \right)
P_{\m}(x;\qq,\tt)
=
\sum_\r P_{\r}(x;\qq,\tt)
Q_{\r/\m}(y;\qq,\tt).
\end{align}
This can be viewed as the action of $\Pi(x,y)$ on $P_{\m}(x;\qq,\tt)$ and the coefficients of this action are given by $Q_{\r/\m}(y;\qq,\tt)$. We can obtain the action of (\ref{pipsiexpX}) on Macdonald polynomials from the action of $\Pi(x,y)$ by applying to it a suitably chosen evaluation homomorphism $\omega$ which acts on symmetric functions in the alphabet $y$. More precisely, let $\d_r$ be a series in powers of $\qq$, $\tt$ and $q$, we define the evaluation homomorphism associated to $\d$ by
\ba
\omega_{\d}(p_r(y)):=\frac{\qq^r-1}{1-\tt^r}\d_r. 
\ea
If we apply this homomorphism  to a symmetric function $f(y)$, then we write the result using the square parenthesis notation
\ba
f\left[\d_1 \right]:=\omega_\d(f(y)).
\ea
We define $\d'_r$
\begin{align}
\d'_r:=(1-\tt^r)(\qq^r-\tt^{r})q^{-r}\tt^{-r}.
\end{align}
Applying $\omega_{\d'}$ to the Cauchy kernel we obtain $\pi(\psid^-_{-j})$ as written in (\ref{pipsiexpX}):
\begin{align}
\label{eq:omegaPi}
\pi(\psid^-_{-j})=\omega_{\d'}\left(\Pi(x,y)\right).
\end{align}
Then applying $\omega_{\d'}$ to (\ref{eq:skewC}) gives us
\begin{align}
\label{eq:psiQ}
\pi(\psid^-_{-j})
P_{\m}(x;\qq,\tt)
=
\sum_\r Q_{\r/\m}[(1-\tt)(\qq-\tt)q^{-1}\tt^{-1}] P_{\r}(x;\qq,\tt).
\end{align}
Therefore we can identify 
\begin{align}
\label{eq:aQ}
a_{\r/\m}(\qq,\tt)=Q_{\r/\m}[(1-\tt)(\qq-\tt)q^{-1}\tt^{-1}].
\end{align}
On the right hand side we have a specialization of a skew Macdonald polynomial in four variables which can be expressed as a sum of a product of branching coefficients of the Macdonald polynomials (see \cite{MacdBook}). This allows one to write down $a_{\r/\m}(\qq,\tt)$ avoiding the use of the implicit definition (\ref{eq:aPexpand}). 

Let us also remark that in the language of plethystic operations the relation (\ref{eq:psiQ}) tells us that $\pi(\psid^-_{-j})$ acts on Macdonald polynomials by shifting the argument by $(1-\tt)(\qq-\tt)q^{-1}\tt^{-1}$. This action of $\psid^-_{-j}$ on symmetric functions appeared in \cite{DiFK}.

\begin{prop}
\label{prop:Lbranch}
For a skew partition $\r/\s$, let $a_{\r/\s}(\qq,\tt)$ be given by (\ref{eq:aQ}) and for a partition $\n$ let $\e_\n,\eb_\n$ be the eigenvalues (\ref{eq:Mev}) of the two Macdonald operators $E$ and $\Eb$. The coefficients $L_{\a/\l,\b/\m}(u)$ are given by the formula 
\begin{align}
\label{eq:Lskew}
L_{\a/\l,\b/\m}(u)
&=
\frac{1}{\e_\a-u \eb_\b} \frac{(1-\tt)(1-\qq^{-1})  q^{2+|\a|-|\l|} }{(1-q^2)(\eb_\a-\eb_\l)}a_{\a/\l}(\qq,\tt)a_{\b/\m}(\qq,\tt).
\end{align}
\end{prop}
\begin{proof}
We separate the $j=0$ term in the summation in (\ref{eq:DRVOzero}), apply  $\pi$ and use (\ref{eq:EEbar}). After this (\ref{eq:DRVOzero}) becomes 
\begin{align}
\label{eq:Refpsi}
\left(u \bar{E}_y -E_x\right)R(x,y;u) 
= s \sum_{j> 0} q^{j} \pi(\ed_{-j})\otimes  \pi(\psid^-_{-j})\cdot R(x,y;u) . 
\end{align}
We then substitute (\ref{eq:Lcoef}) and use (\ref{eq:EPact})
\begin{multline}
\sum_{\l,\m} L_{\l,\m}(u) \left(u \eb_\m -\e_\l\right) P_\l(x;\qq,\tt)P_\m(y;\qq,\tt) 
=\\ 
s  \sum_{\l,\m} L_{\l,\m}(u) \sum_{j> 0} q^{j} \left(\pi(e_{-j})P_\l(x;\qq,\tt) \right)\left( \pi(\psi^-_{-j}) P_\m(y;\qq,\tt)\right). 
\end{multline}
On the right hand side we use (\ref{eq:psimpiact}) and (\ref{eq:epiact}) 
\begin{multline}
\sum_{\l,\m} L_{\l,\m}(u) \left(u \eb_\m -\e_\l\right) P_\l(x;\qq,\tt)P_\m(y;\qq,\tt) 
=\\ 
s
\sum_{\l,\m} L_{\l,\m}(u) 
\sum_{\a,\b} 
\d_{|\a|,|\b|}
 q^{|\a|-|\l|} 
\frac{a_{\b/\m}(\qq,\tt) a_{\a/\l}(\qq,\tt)}{(1-q^{-2})(\eb_\a-\eb_\l)}
P_\a(x;\qq,\tt) P_\b(y;\qq,\tt). 
\end{multline}
After relabelling the indices on the left hand side $\l,\m\rightarrow \a,\b$, and matching coefficients of the same Macdonald polynomials on the two sides of the equation we get 
\begin{align}
L_{\a,\b}(u)
=
s
\sum_{\l,\m} 
  q^{|\a|-|\l|} 
\frac{ a_{\a/\l}(\qq,\tt)a_{\b/\m}(\qq,\tt)}{(1-q^{-2})(u \eb_\b -\e_\a)(\eb_\a-\eb_\l)}
 L_{\l,\m}(u).
\end{align}
After recalling that $s=\k_1/(1-q_2)$, (\ref{eq:Lskew}) follows. 
\end{proof}


\appendix

\section{Calculation of $\Dt(g)$}
\label{appdelta}
The derivation of $\Dt(g)$ is based on the commutation relations of $\gt$. We summarize the outcome of the computations below:
\begin{align}
&\Dt(e_k)=1\otimes e_k+  \sum_{j=0}^{\infty} 
e_{k+j}\otimes q^{-c(k+j)}  \psi^-_{-j},\\
&\Dt(f_k)=f_k \otimes 1 +  \sum_{j=0}^{\infty} q^{-c(k-j)}
\psi^+_{j} \otimes f_{k-j},\\
\label{eq:KhrA}
&
\Dt(h_r)=
h_{r} \otimes 1+
q^{-c r } \otimes h_r, \\
\label{eq:KhmrA}
&
\Dt(h_{-r})=
1 \otimes h_{-r} + 
h_{-r} \otimes q^{c r}.
\end{align}
In the following subsections we treat each equation separately.


\subsection{Calculation of $\Dt(e_k)$}
From (\ref{eq:De}) we can write $\D'(e_k)$ and using the definition (\ref{eq:Dt}) we have 
\begin{align}
\Dt(e_k)=\MK^{-1}\left(
 \sum_{j=0}^{\infty}  q^{c k } \psi_j^{+}\otimes e_{k-j} \right)  \MK+ 
 \MK^{-1} \left(e_k\otimes 1\right) \MK,
\end{align}
where we remind the reader that $\MK$ is defined in \eqref{MK}. Let us compute the two terms separately. 


\subsubsection{Calculation of the first term} The summation in the first term can be represented as follows
\ba
 \sum_{j\geq 0}
q^{c k}
\psi_j^{+}\otimes e_{k-j}
=
&
 q^{ c k -\cp} \sum_{j=0}^{\infty}
\sum_{m=0}^{j}
 \frac{1}{m!}
\sum_{\substack{r_1,\dots,r_m >0\\r_1+\cdots + r_m=j}}
\left(
\prod_{i=1}^m r_i \k_{r_i} 
h_{r_i} \right)
\otimes  [h_{-r_m},[\dots,[h_{-r_1},e_k]\dots]],
\ea
which can be verified using (\ref{psih}) and for the nested commutator one needs to use  the second line of (\ref{eq:hh}). The right hand side can be rewritten using the Baker--Campbell--Hausdorff (BCH) formula, therefore we have
\ba
 \sum_{j=0}^{\infty} \left(
q^{c k}
\psi_j^{+}\otimes e_{k-j} \right)
&=
\exp\left(\sum_{r>0} r \k_r h_r\otimes h_{-r}\right)
\left(q^{ c k -\cp} \otimes e_k\right)
\exp\left(-\sum_{r>0} r \k_r h_r\otimes h_{-r}\right).
\ea
With this representation and using the expression (\ref{MK}) we compute
\ba
\MK^{-1}  \left(\sum_{j=0}^{\infty} q^{c k } \psi_j^{+}\otimes e_{k-j}\right)  \MK 
=
q^{c\otimes d+d\otimes c+\cp\otimes d^{\perp}+d^{\perp}\otimes \cp }
\left(q^{ c k -\cp} \otimes e_k\right)
q^{-c\otimes d-d\otimes c -\cp\otimes d^{\perp}-d^{\perp}\otimes \cp }.
\ea
Using (\ref{degdef})-(\ref{degsefh}) we compute the conjugation with grading operators and find
\begin{align}
\label{Dtek1}
\MK^{-1}  \left(\sum_{j=0}^{\infty}q^{c k } \psi_j^{+}\otimes e_{k-j} \right)  \MK \nonumber=1\otimes e_k.
\end{align}

\subsubsection{Calculation of the second term}
\label{eq:subcalc}
Let us conjugate $e_k\otimes 1$ with the first factor of $\MK$
\ba
&\exp\left(- \sum_{r\geq 1}r \k_r h_{r} \otimes h_{-r} \right)
\left(e_k\otimes 1\right)
\exp\left( \sum_{r\geq 1}r \k_r h_{r} \otimes h_{-r} \right)\\
=
&
 \sum_{j=0}^{\infty}
\sum_{m=0}^{j}
 \frac{(-1)^m}{m!}
\sum_{\substack{r_1,\dots,r_m >0\\r_1+\cdots + r_m=j}}
 [h_{r_m},[\dots,[h_{r_1},e_k]\dots]]
\otimes 
\left(
\prod_{i=1}^m r_i \k_{r_i} 
h_{-r_i} \right)=
 \sum_{j=0}^{\infty} q^{-c j }e_{k+j}\otimes q^{-\cp} \psi^-_{-j},
\ea
where the last equality follows from (\ref{psih}) and  the second line of (\ref{eq:hh}). With this result and $\MK$ given in  (\ref{MK}) we calculate 
\ba
\MK^{-1}  \left( e_k\otimes 1\right)\MK 
= 
 \sum_{j=0}^{\infty} 
q^{c\otimes d+d\otimes c+\cp\otimes d^{\perp}+d^{\perp}\otimes \cp }
\left( q^{-c j }e_{k+j}\otimes q^{-\cp}  \psi^-_{-j}\right)q^{-c\otimes d-d\otimes c -\cp\otimes d^{\perp}-d^{\perp}\otimes \cp} .
\ea
After computing the conjugation with grading operators we find
\begin{align}
\MK^{-1}  \left(e_k\otimes 1\right)  \MK 
 =
  \sum_{j=0}^{\infty}  
e_{k+j}\otimes q^{-c (k+j) }\psi^-_{-j}.
\end{align}


\subsection{Calculation of $\Dt(f_k)$}
From (\ref{eq:De}) we can write $\D'(f_k)$ and using the definition (\ref{eq:Dt}) we have 
\begin{align}
\Dt(f_k)=\MK^{-1} \left(\sum_{j=0}^{\infty} f_{k+j} \otimes q^{c k }  \psi_{-j}^{-} \right) \MK
 + \MK^{-1}  \left(1\otimes  f_k\right) \MK.
\end{align}
Once again we compute the two terms separately. 
\subsubsection{Calculation of the first term} The summation in the first term can be represented as follows
\ba
\sum_{j=0}^{\infty} f_{k+j} \otimes q^{c k }   \psi_{-j}^{-}
=
&
 \sum_{j=0}^{\infty}
\sum_{m=0}^{j}
 \frac{1}{m!}
\sum_{\substack{r_1,\dots,r_m >0\\r_1+\cdots + r_m=j}}
 [h_{r_m},[\dots,[h_{r_1},f_k]\dots]]
\otimes q^{ c k +\cp} \left(
\prod_{i=1}^m r_i \k_{r_i} 
h_{-r_i} \right).
\ea
which can be verified using (\ref{psih}) and for the nested commutator one needs to use  the third line of (\ref{eq:hh}). The right hand side can be rewritten using the BCH formula:
\ba
\sum_{j=0}^{\infty}  f_{k+j} \otimes q^{ c k} \psi_{-j}^{-}
&=
\exp\left(\sum_{r>0} r \k_r h_r\otimes h_{-r}\right)
\left(f_k\otimes q^{ c k +\cp} \right)
\exp\left(-\sum_{r>0} r \k_r h_r\otimes h_{-r}\right).
\ea
With this representation and using the expression (\ref{MK}) we compute
\ba
\MK^{-1}  \left(
\sum_{j=0}^{\infty} f_{k+j} \otimes q^{ c k } \psi_{-j}^{-} \right)\MK 
=
q^{c\otimes d+d\otimes c+\cp\otimes d^{\perp}+d^{\perp}\otimes \cp }
\left(f_k\otimes q^{ c k +\cp} \right)
q^{-c\otimes d-d\otimes c -\cp\otimes d^{\perp}-d^{\perp}\otimes \cp }.
\ea
Using (\ref{degdef})-(\ref{degsefh}) we compute the conjugation with grading operators and find
\begin{align}
\label{Dtfk1}
\MK^{-1}  \left(
\sum_{j=0}^{\infty} f_{k+j} \otimes q^{ c k } \psi_{-j}^{-} \right)\MK 
=f_k \otimes 1.
\end{align}

\subsubsection{Calculation of the second term}
Let us conjugate $1\otimes f_k$ with the first factor of $\MK$
\ba
&\exp\left(- \sum_{r\geq 1}r \k_r h_{r} \otimes h_{-r} \right)
\left(1\otimes f_k\right)
\exp\left( \sum_{r\geq 1}r \k_r h_{r} \otimes h_{-r} \right)\\
=
&
 \sum_{j=0}^{\infty}
\sum_{m=0}^{j}
 \frac{(-1)^m}{m!}
\sum_{\substack{r_1,\dots,r_m >0\\r_1+\cdots + r_m=j}}
\left(
\prod_{i=1}^m r_i \k_{r_i} 
h_{r_i} \right)
\otimes 
 [h_{-r_m},[\dots,[h_{-r_1},f_k]\dots]]
=
 \sum_{j=0}^{\infty} q^{\cp} \psi^+_{j} \otimes q^{c j} f_{k-j},
\ea
where the last equality follows from (\ref{psih}) and  the third line of (\ref{eq:hh}). With this result and $\MK$ given in  (\ref{MK}) we calculate 
\ba
\MK^{-1}  \left(1\otimes f_k\right)  \MK 
= 
 \sum_{j=0}^{\infty} 
q^{c\otimes d+d\otimes c+\cp\otimes d^{\perp}+d^{\perp}\otimes \cp }
\left(  q^{\cp}\psi^+_{j} \otimes q^{c j} f_{k-j}\right)q^{-c\otimes d-d\otimes c -\cp\otimes d^{\perp}-d^{\perp}\otimes \cp} .
\ea
After computing the conjugation with grading operators we find
\begin{align}
\MK^{-1}  \left(1\otimes f_k\right)  \MK 
 =
  \sum_{j=0}^{\infty} q^{-c(k-j)} 
\psi^+_{j} \otimes f_{k-j}.
\end{align}


\subsection{Calculation of $\Dt(h_r)$}
From \eqref{eq:De} we need to calculate: 
\begin{align*}
\Dt(h_r) &=\MK^{-1} \left(1\otimes h_r\right) \MK
 +  \MK^{-1}   \left(h_r \otimes q^{-c r }\right)  \MK, \\
\Dt(h_{-r}) &= \MK^{-1}  \left( q^{c r }\otimes h_{-r}\right)  \MK
 + \MK^{-1}  \left(h_{-r}\otimes 1 \right) \MK.
\end{align*} 
As before we conjugate  $h_r\otimes 1$ and $1\otimes h_r$ with factors of $\MK$. We have the following identities
\begin{align*}   
%
q^{c\otimes d+d\otimes c+\cp\otimes d^{\perp}+d^{\perp}\otimes \cp }
\left(h_{r} \otimes 1 \right)
q^{-c\otimes d-d\otimes c -\cp\otimes d^{\perp}-d^{\perp}\otimes \cp} 
&=
h_{r} \otimes q^{-c r },\\
%
q^{c\otimes d+d\otimes c+\cp\otimes d^{\perp}+d^{\perp}\otimes \cp }
\left(1 \otimes h_{r} \right)
q^{-c\otimes d-d\otimes c -\cp\otimes d^{\perp}-d^{\perp}\otimes \cp} 
&=
q^{-c r } \otimes h_r ,
\end{align*}
which are valid for $r>0$ and $r<0$. Next, using the BCH formula we compute
\begin{align*}
\exp\left( -\sum_{s\geq 1}s \k_s h_{s} \otimes h_{-s} \right)
\left(
h_r \otimes 1
\right)
\exp\left( \sum_{s\geq 1}s \k_s h_{s} \otimes h_{-s} \right)
&=h_r \otimes 1,\\
\exp\left( -\sum_{s\geq 1}s \k_s h_{s} \otimes h_{-s} \right)
\left(
1 \otimes h_{-r}
\right)
\exp\left( \sum_{s\geq 1}s \k_s h_{s} \otimes h_{-s} \right)
&=1 \otimes h_{-r},\\
%
%
\exp\left( -\sum_{s\geq 1}s \k_s h_{s} \otimes h_{-s} \right)
\left(
h_{-r} \otimes 1
\right)
\exp\left( \sum_{s\geq 1}s \k_s h_{s} \otimes h_{-s} \right)
&=
h_{-r} \otimes 1 
+
(q^{-c r }-q^{c r }) \otimes h_{-r},\\
%
%
\exp\left( -\sum_{s\geq 1}s \k_s h_{s} \otimes h_{-s} \right)
\left(
1 \otimes h_r
\right)
\exp\left( \sum_{s\geq 1}s \k_s h_{s} \otimes h_{-s} \right)
&=
1 \otimes h_r
+
h_r \otimes (q^{c r }-q^{-c r }) .
\end{align*}   
Putting these results together we get
\begin{align}
\label{eq:Kh1}
\begin{split}
\MK^{-1} \left(1\otimes h_r\right) \MK &= \left(q^{-c r }  \otimes h_r +
h_{r} \otimes (q^{c r}-q^{-c r}) q^{-c r } \right), \\
 \MK^{-1}  \left( h_r \otimes q^{-c r } \right)  \MK &=  h_{r} \otimes q^{-2 c r }, \\
 \MK^{-1}  \left(  1\otimes h_{-r}\right)  \MK &=
q^{2 c r  } \otimes h_{-r}  ,\\
\MK^{-1}  \left(h_{-r}\otimes 1\right)  \MK &= 
h_{-r} \otimes q^{c r} 
+
q^{c r}(q^{-c r }-q^{c r })  \otimes h_{-r}.
\end{split} 
\end{align}
Adding the first  two equations in (\ref{eq:Kh1}) leads to (\ref{eq:KhrA}) and adding the last two equations in (\ref{eq:Kh1}) leads to (\ref{eq:KhmrA}).


\section{Proof of Proposition~\ref{lem:P}}
\label{ap:proofProp:P}
In this appendix we prove Proposition~\ref{lem:P}, namely that the coefficients $(\Rb(u))_{\m,\r}^{\n,\s}$ of the $R$-matrix $\Rb(u)$ as defined in \eqref{eq:Rbexp} depend only on two partitions. We furthermore give an explicit generic form for the $R$-matrix that manifestly commutes with $\Dt(a_r)$.

Fix $r>0$ and insert (\ref{eq:Rbexp}) into (\ref{eq:DRa})
\begin{align}
\label{eq:Rbcom}
\sum_{\m,\r,\n,\s\in \MP}
(\Rb(u))_{\m,\r}^{\n,\s}
\left(
[a^*_{\m},a_r] a_{\n}\otimes a^*_{\r}a_{\s} 
+
q^{-r} 
a^*_{\m}a_{\n}\otimes [a^*_{\r},a_r] a_{\s} \right)
=0.
\end{align}
In the two commutators of the form $[a^*_{\m},a_r]$, the part of $a^*_{\m}$ that does not commute with $a_r$ equals to $a^*_{(r^k)}$, where $k=m_r(\m)$. We compute the commutator
\begin{align*}
[a^*_{(r^k)},a_r]=-q^r\frac{\k_r}{r (1-q_2^r)}a^*_{(r^{k-1})},
\end{align*}
plug this into (\ref{eq:Rbcom}) and cancel the common factor. This gives
\begin{align}
\label{eq:Rbcom2}
\sum_{\m,\r,\n,\s\in \MP}
(\Rb(u))_{\m,\r}^{\n,\s}
\left(
a^*_{\m\psub (r)} a_{\n}\otimes a^*_{\r}a_{\s} 
+
q^{-r} 
a^*_{\m}a_{\n}\otimes a^*_{\r\psub (r)} a_{\s} \right)
=0.
\end{align}   
For fixed $r$, the first term in \eqref{eq:Rbcom2} is summed over all $\m$ which have at least one part equal to $r$ and the second term is summed over all $\r$ which have at least one part equal to $r$. Because of this we can shift these summation variables, i.e. $\m\rightarrow \m\padd  (r)$ and $\r\rightarrow \r\padd (r)$, to get 
\begin{align*}
\sum_{\m,\r,\n,\s\in \MP}
\left(
(\Rb(u))_{\m\padd  (r),\r}^{\n,\s}
+
q^{-r} 
(\Rb(u))_{\m,\r\padd (r)}^{\n,\s}\right)
a^*_{\m}a_{\n}\otimes a^*_{\r} a_{\s} 
=0.
\end{align*}   
This equation must be satisfied for all summed partitions leading us to
\begin{align*}
(\Rb(u))_{\m\padd (r),\r}^{\n,\s}
+
q^{-r} 
(\Rb(u))_{\m,\r\padd (r)}^{\n,\s}
=0.
\end{align*}
After changing the labels we obtain
\begin{align*}
(\Rb(u))_{\m,\r}^{\n,\s}
=- q^{r} 
(\Rb(u))_{\m\padd (r),\r\psub (r)}^{\n,\s},
\end{align*}
which we can iteratively use for all $r\in \r$, i.e. 
\begin{align}
\label{eq:Rshiftx}
(\Rb(u))_{\m,\r}^{\n,\s}
=(-1)^{\ell(\r)} q^{|\r|} 
(\Rb(u))_{\m\padd \r,\pzero}^{\n,\s}.
\end{align}
In a similar manner we compute the commutation relation (\ref{eq:DRadag}) and find
\begin{align}
\label{eq:Rmshiftx}
(\Rb(u))_{\m,\r}^{\n,\s}
=(-1)^{\ell(\s)} q^{|\s|} 
(\Rb(u))_{\m, \r}^{\n\padd \s,\pzero}.
\end{align}
Combining (\ref{eq:Rshiftx}) with (\ref{eq:Rmshiftx}) we arrive at 
\ba
(\Rb(u))_{\m,\r}^{\n,\s}
=(-1)^{\ell(\r)+\ell(\s)} q^{|\r|+|\s|} 
P_{\m\padd \r,\n\padd \s}(u),
\ea
as claimed.

\section{Reduction of the coproduct relation}
\label{app:currents}
In this Appendix we collect the derivations which allow us to reduce the coproduct equations for the $R$-matrix to the vector form stated in Proposition \ref{prop:copef}. For convenience we will sometimes write the position of a current or vertex operator in a tensor product using an index
\ba
e(z)\otimes 1= e_1(z),\qquad 1\otimes e(z)=e_2(z). 
\ea 
Note that this index should not be confused with that of the modes of $e(z)$.

\subsection{The coproduct equation with $\D(e(z))$ and $\D(f(z))$ and vertex operators}
Our derivations are given using the language of currents and vertex operators. Therefore we first present some useful formulas. 

Let us write the coproduct (\ref{eq:De}) and (\ref{eq:RDe}) for the modes $e_k$ and $f_k$ using the current generators. This is done by summing these equations over $k\in \mathbb{Z}$ with $z^{-k}$. We write $\D(e(z))$, $\D(f(z))$ and $\Dt(e(z))$ and $\Dt(f(z))$ 
\begin{align}
\label{eq:copce}
\begin{split}
\D(e(z)) &= e(z q^{-c_2}) \otimes \psi^{+}(z q^{-c_2}) + 1\otimes e(z), \\
\D(f(z)) &=  f(z)\otimes 1 + \psi^{-}(z q^{-c_1})\otimes f(z q^{-c_1}), \\
\Dt(e(z)) &= e(z q^{c_2}) \otimes \psi^{-}(z) + 1\otimes e(z), \\
\Dt(f(z)) &= f(z)\otimes 1 + \psi^+(z)\otimes f(z q^{c_1}).
\end{split}
\end{align}   
Here we recall that the index $i$ of $c_i$ indicates which factor of the tensor product the operator $c$ is acting on.
In the Fock representation $c_1=c_2=1$ (and $\cp=0$). Recall the vertex operators (\ref{VOe}) and (\ref{psiexp}) and write them using the Heisenberg generators $a_{\pm r}$
\begin{align}
\label{VOea}
\begin{split}
e(z) &= \frac{1-q_2}{\k_1} u ~
\exp\left(\sum_{r=1}^{\infty}q^{-r}a_{-r}z^{ r}\right)
\exp\left(-\sum_{r=1}^{\infty}a_{ r}z^{- r}\right),\\
f(z) &= -\frac{1-q_2}{q_2 \k_1}u^{-1} 
\exp\left(-\sum_{r=1}^{\infty}a_{- r}z^{ r}\right)
\exp\left(\sum_{r=1}^{\infty}q^r a_{ r}z^{- r}\right),\\
\psi^{\pm}(z) &= \exp\left(\pm\sum_{r=1}^{\infty} (q^{r}-q^{-r})a_{\pm r} z^{\mp r}\right). 
\end{split}
\end{align}
Each of the entries of the coproducts in (\ref{eq:copce}) can be factorised according to
\ba
A\otimes B = (A\otimes 1)(1\otimes B)=(1\otimes B)(A\otimes 1).
\ea
Because of this we can replace the Heisenberg operators $a_{\pm r}$ inside the exponentials by $a_{\pm r}\otimes 1$ and $1\otimes a_{\pm r}$, thus absorbing the sign of the tensor product into the exponentials of the vertex operators. After this we use the formulas (\ref{eq:a1tobc}) and (\ref{eq:a2tobc}) in order to rewrite the vertex operators (\ref{VOea}) in terms of $b_{\pm r}$ and $c_{\pm r}$. 

Introduce a pair of vertex operators 
\ba
\phi^{\pm}(A;z):= 
\exp\left(\sum_{r=1}^{\infty}\frac{1}{q^r+q^{-r}}A_{\pm r}z^{\mp r}\right).
\ea
where $A_{\pm r}$ is either $b_{\pm r}$ or $c_{\pm r}$. Their ordering relation is given by 
\begin{align}
\label{eq:phicr}
&\phi^+(A;z)
\phi^-(A;w)
= 
\exp\left(\sum_{r>0} \frac{(1-q_1^r)(1-q_3^r)}{r(q^r+q^{-r})}w^r z^{-r}\right)\,
\phi^-(A;w)\phi^+(A;z).
\end{align}
We have 
\begin{align}
\label{eq:VOef}
\begin{split}
e_1(z)
&=\frac{1-q_2}{\k_1} u_1 ~ \phi^-(b;z)\phi^+(b;z q^{-1})^{-1} \phi^-(c;z q^{-1})\phi^+(c;z)^{-1},\\
e_2(z)
&=\frac{1-q_2}{\k_1} u_2 ~ \phi^-(b;z q^{-1})\phi^+(b;z )^{-1}  \phi^-(c;z)^{-1}\phi^+(c;z q^{-1}),\\
f_1(z)
&=-\frac{1-q_2}{q_2 \k_1}u_1^{-1}  \phi^-(b;z q)^{-1}\phi^+(b;z q^{-2}) \phi^-(c;z)^{-1}\phi^+(c;z q^{-1}),\\
f_2(z)
&=-\frac{1-q_2}{q_2 \k_1}u_2^{-1}   \phi^-(b;z)^{-1} \phi^+(b;z q^{-1}) \phi^-(c;z q)\phi^+(c;z q^{-2})^{-1},
\end{split}
\end{align}
and
\begin{align}
\label{eq:VOpsi}
\begin{split}
\psi_1^{+}(z)
&= \phi^+(b;z)^{-1} \phi^+(b;z q^{-2}) 
\phi^+(c;z q)^{-1} \phi^+(c;z q^{-1}),\\
\psi_2^{+}(z)
&=
\phi^+(b;z q)^{-1} \phi^+(b;z q^{-1}) 
\phi^+(c;z) \phi^+(c;z q^{-2})^{-1},\\
\psi_1^{-}(z)
&=
 \phi^-(b;z) \phi^-(b;z q^2)^{-1} 
\phi^-(c;z q^{-1}) \phi^-(c;z q)^{-1},\\
\psi_2^{-}(z)
&=
 \phi^-(b;z q^{-1}) \phi^-(b;z q)^{-1} 
\phi^-(c;z)^{-1} \phi^-(c;z q^{2}).
\end{split}
\end{align}

\subsection{Proof of Proposition \ref{propRFE}}
\label{sec:VOeq}
Let us take the coproduct equation (\ref{eq:RDt}) for the currents $e(z)$ and $f(z)$ using (\ref{eq:copce})
\begin{gather*}
\MRb
\left(e(z q^{-c_2}) \otimes \psi^{+}(z q^{-c_2}) + 1\otimes e(z) \right)
=
\left(e(z q^{c_2}) \otimes \psi^{-}(z) + 1\otimes e(z)\right) \MRb,\\
\MRb
\left( f(z)\otimes 1 + \psi^{-}(z q^{-c_1})\otimes f(z q^{-c_1})\right)
=
\left( f(z)\otimes 1 + \psi^+(z)\otimes f(z q^{c_1})\right)  \MRb,
\end{gather*} 
and specialize the Fock representation $V(u_1)\otimes V(u_2)$ where we set $u_1=1$ and $u_2=u$,
\begin{gather*}
\Rb(u)
\left(e(z q^{-1}) \otimes \psi^{+}(z q^{-1}) + 1\otimes e(z) \right)
=
\left( e(z q) \otimes \psi^{-}(z) + 1\otimes e(z)\right) \Rb(u),\\
\Rb(u)
\left( f(z)\otimes 1 + \psi^{-}(z q^{-1})\otimes f(z q^{-1})\right)
=
\left( f(z)\otimes 1 + \psi^+(z)\otimes f(z q)\right)  \Rb(u).
\end{gather*} 
Replacing the tensor product notation with the index notation we write
\begin{gather}
\label{eq:RA}
[\Rb(u),e_2(z)]
=
 e_1(z q)  \psi_2^{-}(z) \Rb(u)-\Rb(u) e_1(z q^{-1})  \psi_2^{+}(z q^{-1}) ,\\
 \label{eq:RB}
[\Rb(u), f_1(z)]
=  \psi_1^+(z)f_2(z q)\Rb(u)- 
\Rb(u) \psi_1^{-}(z q^{-1}) f_2(z q^{-1}).
\end{gather} 

The operator $\Rb(u)$ in the Fock representation depends only on the Heisenberg operators $c_{\pm r}$ (\ref{eq:Rbarbc}). We can insert the operators \eqref{eq:VOef} and \eqref{eq:VOpsi} into (\ref{eq:RA}) and (\ref{eq:RB}). The latter two lead to the same equation hence we only focus on (\ref{eq:RA}), which becomes
\begin{multline*}
u \,
\phi^-(b;z q^{-1})\phi^+(b;z )^{-1}
[\Rb(u),\phi^-(c;z)^{-1}\phi^+(c;z q^{-1})]\\
=
\phi^-(b;z q)\phi^+(b;z )^{-1}  \phi^-(b;z q^{-1}) \phi^-(b;z q)^{-1} 
 \phi^-(c;z )\phi^+(c;z q)^{-1} \phi^-(c;z)^{-1} \phi^-(c;z q^{2}) \Rb(u)\\
-
\Rb(u) 
 \phi^-(b;z q^{-1}) \phi^+(b;z)^{-1}
 \phi^-(c;z q^{-2}) \phi^+(c;z q^{-3})^{-1}.
\end{multline*}
We can commute the operators using (\ref{eq:phicr}) and find that on both sides the parts that depend on $b_{\pm r}$ cancel and we arrive at
\ba
u 
[\Rb(u),\phi^-(c;z)^{-1}\phi^+(c;z q^{-1})]
&=
  \phi^-(c;z q^{2})  \phi^+(c;z q)^{-1}\Rb(u)- \Rb(u)   \phi^-(c;z q^{-2}) \phi^+(c;z q^{-3})^{-1}.
\ea
This results in equation (\ref{eq:RFE}) which is stated in Proposition \ref{propRFE}.  We also stated this equation in the introduction in (\ref{eq:VOR}), where the operators $\vphi^\pm(z)$ used there read
\ba
\vphi^+(z)=\phi^-(c;z)^{-1}\phi^+(c;z q^{-1}), \qquad 
\vphi^-(z)=\phi^-(c;z)  \phi^+(c;z q^{-1})^{-1}. 
\ea

\subsection{Proof of Proposition \ref{prop:copef}}
\label{ap:proofcop}
In this subsection we show how to derive (\ref{eq:DRVO}) and (\ref{eq:DRVOd})  thus proving Proposition \ref{prop:copef}.  
Consider the vertex operator equation (\ref{eq:RFE}). Applying $\iota$, defined in (\ref{eq:iso1}), and using Lemma \ref{lem:phi} we rewrite the four terms of this equation in $V\otimes V$
\ba
\iota\left(\Rb(u)  \phi^-(z)^{-1} \phi^+(z q^{-1}) \right)&= 
\left(\phi^-(z)^{-1}\otimes \phi^-(z^{-1} q)^{-1} \phi^+(z^{-1}) \right)\DR(u) 
,\\
\iota\left(\phi^-(z)^{-1} \phi^+(z q^{-1})  \Rb(u)\right) & =
\left(\phi^-(z)^{-1}\phi^+(z q^{-1}) \otimes  \phi^-(z^{-1} q)^{-1} \right)\DR(u) 
,\\
\iota\left(\phi^-(z q^{2})  \phi^+(z q)^{-1}   \Rb(u)\right) & =
\left(\phi^-(z q^2)\phi^+(z q)^{-1} \otimes  \phi^-(z^{-1} q^{-1}) \right)\DR(u) 
,\\
\iota\left( \Rb(u)   \phi^-(z q^{-2})  \phi^+(z q^{-3})^{-1}\right)& =
\left(\phi^-(z q^{-2})\otimes \phi^-(z^{-1} q^{3})  \phi^+(z^{-1} q^{2})^{-1} \right)\DR(u) 
 .
\ea
Putting these four terms together we rewrite (\ref{eq:RFE}) as
\begin{multline}
\label{eq:DRphiX}
u \left(\phi^-(z)^{-1}\otimes \phi^-(z^{-1} q)^{-1} \phi^+(z^{-1}) - 
\phi^-(z)^{-1}\phi^+(z q^{-1}) \otimes  \phi^-(z^{-1} q)^{-1} \right)\DR(u) 
\\
=
\left(\phi^-(z q^2)\phi^+(z q)^{-1} \otimes  \phi^-(z^{-1} q^{-1})-
\phi^-(z q^{-2})\otimes \phi^-(z^{-1} q^{3})  \phi^+(z^{-1} q^{2})^{-1} \right)\DR(u). 
\end{multline}
The next step is to multiply (\ref{eq:DRphiX}) from the left by the factor
\ba
\phi^-(z)\otimes \phi^-(z^{-1} q^{3})^{-1},
\ea
which leads us to 
\begin{multline}
\label{eq:DRphiY2}
u \left(
1 \otimes  \phi^-(z^{-1} q^{3})^{-1} \phi^-(z^{-1} q)^{-1} \phi^+(z^{-1})- 
\phi^+(z q^{-1}) \otimes   \phi^-(z^{-1} q^{3})^{-1} \phi^-(z^{-1} q)^{-1} \right)\DR(u) 
\\
=
\left(\phi^-(z)\phi^-(z q^2)\phi^+(z q)^{-1} \otimes 
\phi^-(z^{-1} q^{3})^{-1} \phi^-(z^{-1} q^{-1})-
\phi^-(z)\phi^-(z q^{-2})\otimes 
 \phi^+(z^{-1} q^{2})^{-1} \right)\DR(u).
\end{multline}

Now we can take the formulas of the vertex operator representation (\ref{VOea}) and write them in terms of the Heisenberg algebra $\MH_c$ using the relation
\begin{align*}
\MH_a \rightarrow \MH_c:
\qquad
a_{-r} \mapsto q^r c_{-r},
\quad
a_{r} \mapsto \frac{1}{q^r+q^{-r}} c_r,\qquad r>0.
\end{align*}
We obtain 
\begin{align*}
\ed(z) &= s^{-1}v
\exp\left(\sum_{r=1}^{\infty}c_{-r}z^{ r}\right)
\exp\left(-\sum_{r=1}^{\infty}\frac{1}{q^r+q^{-r}}c_{ r}z^{- r}\right),\\
\fd(z) &= -s^{-1}q_2^{-1} v^{-1}
\exp\left(-\sum_{r=1}^{\infty}q^r c_{- r}z^{ r}\right)
\exp\left(\sum_{r=1}^{\infty}\frac{q^r}{q^r+q^{-r}} c_{ r}z^{- r}\right),\\
\psid^{+}(z) &= \exp\left(\sum_{r=1}^{\infty} \frac{(q^{r}-q^{-r})}{q^r+q^{-r}}c_{ r} z^{- r}\right),\\
\psid^{-}(z) &= \exp\left(-\sum_{r=1}^{\infty} (q^{r}-q^{-r})q^r c_{- r} z^{ r}\right),
\end{align*}
where $v$ plays the role of the spectral parameter and $s=(1-q_1)(1-q_3)$. 
We can write these operators using the vertex operators $\phi^\pm(z)$
\begin{align}
\label{VOefphiX}
\begin{split}
\ed(z) &= s^{-1} v ~
\phi^-(z q^{-1})\phi^-(z q)
\phi^+(z)^{-1},\\
\fd(z) &= -s^{-1} q_2^{-1}v^{-1} 
\phi^-(z q^2)^{-1}\phi^-(z)^{-1}
\phi^+(z q^{-1}),\\
\psid^{+}(z) &=
\phi^+(z q)\phi^+(z q^{-1})^{-1},\\
\psid^{-}(z) &=
\phi^-(z q^3)^{-1}\phi^-(z q^{-1} ).
\end{split}
\end{align}
Consider two copies of these operators acting on $V(1)\otimes V(1)$. Then we can match the first term on the left hand side in (\ref{eq:DRphiY2}) with the current $\fd$ in (\ref{VOefphiX})
\ba
1\otimes  \phi^-(z^{-1} q^{3})^{-1} \phi^-(z^{-1} q)^{-1} \phi^+(z^{-1}) =
-s q^2\otimes  \fd(z^{-1} q),
\ea
and the first term on the right hand side with a product of the currents $\ed$ and $\psid^-$ in (\ref{VOefphiX})
\ba
\phi^-(z)\phi^-(z q^2)\phi^+(z q)^{-1} \otimes 
\phi^-(z^{-1} q^{3})^{-1} \phi^-(z^{-1} q^{-1})&=
s\, \ed(z q)\otimes \psid^{-}(z^{-1}).
\ea
Rewriting accordingly (\ref{eq:DRphiY2})  we arrive at 
\begin{align}
\label{eq:DRphiY3}
x&\left(
-s q^2 \otimes  f(z^{-1} q) - 
\phi^+(z q^{-1}) \otimes   \phi^-(z^{-1} q^{3})^{-1} \phi^-(z^{-1} q)^{-1} \right)\DR(u) 
\\
=
&\left(
s~e(z q)\otimes \psi^{-}(z^{-1})
-
\phi^-(z)\phi^-(z q^{-2})\otimes 
 \phi^+(z^{-1} q^{2})^{-1} \right)\DR(u). \nonumber 
\end{align}
The remaining two terms in (\ref{eq:DRphiY3}) are matched with $\Phi^{\pm}(z)$ given in (\ref{eq:Phip}) and (\ref{eq:Phim}). After this the equation (\ref{eq:DRVO}) follows. In order to obtain (\ref{eq:DRVOd}) we start with (\ref{eq:DRphiX}) and multiply it by 
\ba
\phi^-(z q^2)^{-1}\otimes \phi^-(z^{-1} q).
\ea
The resulting equation can be matched with (\ref{eq:DRVOd}) using (\ref{VOefphiX}) and (\ref{eq:Psipm}).

\section*{Acknowledgments}
We gratefully acknowledge support from the Australian Research Council Centre of Excellence for Mathematical and Statistical Frontiers (ACEMS). This project was supported in part by the ARC Discovery Grant DP190102897. 
A. G. would like to thank Jean-Emile Bourgine, Stephane Dartois, Matteo Mucciconi and Ole Warnaar for useful comments and discussions. We thank Andrei Okounkov and Olivier Schiffmann for comments and references and Andrei Negu{\c{t}} and Andrey Smirnov for explaining their works.

\end{document}